\documentclass[11pt]{article}
\usepackage[T1]{fontenc}
\usepackage{lmodern}
\usepackage{fullpage}
\usepackage{amsmath,amsthm,amsfonts,amssymb,bbm}
\usepackage[authoryear]{natbib}
\usepackage{graphicx,caption,subcaption}
\usepackage{algorithm}
\usepackage{algpseudocode}
\usepackage[margin=2em]{caption}
\usepackage{cleveref}
\usepackage{color}
\usepackage{hyperref}
\hypersetup{
    pdfauthor={Yiling Chen and Bo Waggoner},
    pdftitle={Informational Substitutes},
    colorlinks=true,
    citecolor=blue,
    linkcolor=blue,
    urlcolor=blue,
    linktoc=page
}

\makeatletter
\g@addto@macro\normalsize{%
  \setlength\abovedisplayskip{0.5em}
  \setlength\belowdisplayskip{0.5em}
  \setlength\abovedisplayshortskip{0.1em}
  \setlength\belowdisplayshortskip{0.5em}
}
\makeatother

\newcommand{\Omit}[1]{}

\theoremstyle{plain}
\newtheorem{theorem}{Theorem}[subsection]
\newtheorem*{theorem*}{Theorem}
\newtheorem{fact}{Fact}[subsection]
\newtheorem*{fact*}{Fact}
\newtheorem{observation}{Observation}[subsection]
\newtheorem*{observation*}{Observation}
\newtheorem{lemma}{Lemma}[subsection]
\newtheorem*{lemma*}{Lemma}
\newtheorem{proposition}{Proposition}[subsection]
\newtheorem*{proposition*}{Proposition}
\newtheorem{corollary}{Corollary}[subsection]
\newtheorem*{corollary*}{Corollary}
\newtheorem{claim}{Claim}[subsection]
\newtheorem*{claim*}{Claim}
\theoremstyle{definition}
\newtheorem{definition}{Definition}[subsection]
\newtheorem*{definition*}{Definition}
\theoremstyle{remark}
\newtheorem{example}{Example}[subsection]
\newtheorem*{example*}{Example}

\newcommand{\R}{\mathbb{R}}  
\DeclareMathOperator*{\E}{\mathbb{E}}   
\DeclareMathOperator*{\bigtimes}{\mathrel{\text{\scalebox{1.3}{$\times$}}}}
\newcommand{\smallwedge}{\mathrel{\text{\raisebox{0.25ex}{\scalebox{0.8}{$\wedge$}}}}}
\newcommand{\smallvee}{\mathrel{\text{\raisebox{0.25ex}{\scalebox{0.8}{$\vee$}}}}}

\newcommand{\D}{\mathcal{D}}  
\newcommand{\V}{\mathcal{V}}  
\newcommand{\Lat}{\mathcal{L}}  
\newcommand{\meet}{\smallwedge}
\newcommand{\join}{\smallvee}

\DeclareMathOperator*{\bigjoin}{\bigvee}

\def\sigsel/{\texorpdfstring{\textsc{SignalSelection}}{SignalSelection}}

\newcommand{\defsubsupermodular}{
A function $f$ from a lattice to the reals is \emph{submodular} if it exhibits diminishing marginal value: For all $A',A,B$ on the lattice with $A \meet B \preceq A' \preceq A$,
 \[ f(B \join A') - f(A') \geq f(B \join A) - f(A) . \]
It is \emph{supermodular} if it exhibits increasing marginal value: For all $A,A',B$ on the lattice with $A \meet B \preceq A' \preceq A$, the above inequality is reversed.
The sub- or super-modularity is \emph{strict} if, whenever $A$ and $B$ are incomparable on the lattice's ordering and $A' \neq A$, the inequality is strict.
} 

\newcommand{\defsubscompsweak}{
In the context of a decision problem $u$ and prior $P$, the signals $A_1,\ldots,A_n$ are \emph{weak substitutes} if $\V^{u,P}$ is submodular on the subsets signal lattice generated by $A_1,\ldots,A_n$.
They are \emph{weak complements} if $\V^{u,p}$ is supermodular on their subsets signal lattice.
The signals are \emph{strict} substitutes (complements) if the $\V^{u,p}$ is strictly submodular (supermodular).
} 

\newcommand{\defsubscompsmodstrong}{
In the context of a decision problem $u$ and prior $P$, the signals $A_1,\ldots,A_n$ are \emph{moderate (respectively, strong) substitutes} if, for all $A,B$ on the subsets signal lattice, and for all $A'$ on the discrete (respectively, continuous) signal lattice with $A \meet B \preceq A' \preceq A$,
  \[ \V^{u,P}(B \join A') - \V^{u,P}(A') \geq \V^{u,P}(B \join A) - \V^{u,P}(A) . \]
They are \emph{moderate (strong) complements} if, for all $A,B$ on the subsets signal lattice and all $A'$ on their discrete (continuous) signal lattice with $A \meet B \preceq A' \preceq A$, the above inequality is reversed.

The signals are \emph{strict} substitutes (complements) if whenever $A$ and $B$ are incomparable and $A' \neq A$, the respective inequality is strict.
} 

\newcommand{\theoremrevelation}{
For any decision problem $u$, there exists a proper scoring rule $S: \Delta_E \times E \to \mathbb{R}$ that is equivalent to the original decision problem in that for all information structures $P$ and signals $A$, $\V^{S,P}(A) = \V^{u,P}(A)$.
} 

\newcommand{\cordecisioncapturedbyG}{
 For any decision problem $u$ there exists a corresponding convex function $G: \Delta_E \to \R$, and for every such $G$ there exists a decision problem $u$, such that $G(q)$ is the expected utility for acting optimally when the agent's posterior belief on $E$ is $q$.

Hence for instance $\V(A) = \E_a G(p_a)$, where $p_a$ is the posterior on $E$ given $A=a$.
} 

\newcommand{\defsubsG}{
For any decision problem, letting $G$ be the associated expected score function, signals are respectively \emph{(weak, moderate, strong) substitutes} for that decision problem if and only if, for all $A,B$ on the subsets lattice and all $A'$ on the respectively (subsets, discrete, continuous) lattice with $A \meet B \preceq A' \preceq A$,
 \[ \E_{a',b} G(p_{a'b}) - \E_{a'} G(p_{a'}) \geq \E_{a,b} G(p_{ab}) - \E_a G(p_a) . \]
} 

\newcommand{\lemmamarginalcontributionbregman}{
In a decision problem, for any signals $A,B$, the marginal value of $B$ given $A$ is
 \[ \V(A \join B) - \V(A) = \E_{a,b} D_G(p_{ab}, p_a)  \]
where $G$ is the associated expected score function and $p_{ab}$ is the distribution of $E$ given $A=a,B=b$.
} 

\newcommand{\defsubsbregman}{
For any decision problem $u$, let $D_G$ be the Bregman divergence of the corresponding expected score function $G$.
Signals are \emph{(weak, moderate, strong) substitutes} if and only if, for all $A,B$ on the subsets lattice and all $A'$ on the (subsets, discrete, continuous) lattice with $A \meet B \preceq A' \preceq A$,
 \[ \E_{a',b} D_G(p_{a'b}, p_{a'}) \geq \E_{a,b} D_G(p_{ab}, p_a) . \]
} 

\newcommand{\defsubsentropy}{
For any decision problem $u$, let the generalized entropy function $h = -G$ where $G$ is the corresponding expected score function.
Signals are \emph{(weak, moderate, strong) substitutes} for $u$ if and only if, for all $A,B$ on the subsets lattice and all $A'$ on the (subsets, discrete, continuous) lattice with $A \meet B \preceq A' \preceq A$,
 \[ h(E | A') - h(E | A' \join B) \geq h(E | A) - h(E | A \join B) . \]
} 

\newcommand{\thmsubsrushequilibrium}{
If signals are distinguishable and are strict, strong substitutes, then for every set of traders and trading orders, every Bayes-Nash equilibrium is all-rush.
} 

\newcommand{\thmsubsnonrush}{
If signals are distinguishable and are not strong substitutes, then there is a trading order such that no perfect Bayesian equilibrium is all-rush.
} 

\newcommand{\thmcompsdelayequilibrium}{
If signals are distinguishable and are strict, strong complements, then for every set of traders and trading order, every perfect Bayesian equilibrium is an all-delay equilibrium.
} 

\newcommand{\thmcompsnondelay}{
If signals are distinguishable and are not strong complements, then there is a set of traders and trading order such that no perfect Bayesian equilibrium is all-delay.
} 

\newcommand{\thmalgsubspos}{
If signals are weak substitutes, there are polynomial-time $1-1/e - o(1)$ approximations for \sigsel/ under a variety of constraint families in the oracle model or with succinct input as in Propositions \ref{prop:compute-V-small-sparse} or \ref{prop:compute-V-small-oracle}.
These include cardinality constraints (select at most $k$ signals), budget constraints (given a price associated with each signal, select a subset with total price at most $B$), and \emph{matroid} constraints.
} 

\newcommand{\thmftosigsel}{
For every monotone increasing set function $f: 2^N \to \mathbb{R}$, with $|N| = n$, there exists a decision problem and a set of signals $A_1,\dots,A_n$ such that $f(S) = \V(\bigjoin_{i \in S} A_i)$ for all $S \subseteq N$.

Furthermore, in the construction, it is trivial to compute and represent any posterior $p_{a_i : i \in S}$ conditioned on any realizations of any subset $S$ of signals; and trivial (given a single evaluation of $f(S)$) to compute the expected optimal utility $G(p_{a_i : i \in S})$.
} 

\newcommand{\corcomplementsalghard}{
No algorithm for \sigsel/ accepting the oracle model of inputs, or any of the other models discussed, can guarantee a nonzero approximation in general even for cardinality constraints and even if signals are assumed to be weak complements, unless it makes an exponential number of queries to the oracles.
} 

\begin{document}
\pagenumbering{roman}

\title{Informational Substitutes}
\author{Yiling Chen and Bo Waggoner \\
 Harvard, University of Pennsylvania \\
 \texttt{yiling@seas.harvard.edu, bwag@seas.upenn.edu}} 
\date{Updated: March 2017}
\maketitle

\begin{abstract}
We propose definitions of substitutes and complements for pieces of information (``signals'') in the context of a decision or optimization problem, with game-theoretic and algorithmic applications.
In a game-theoretic context, substitutes capture diminishing marginal value of information to a rational decision maker.
We use the definitions to address the question of how and when information is aggregated in prediction markets.
Substitutes characterize ``best-possible'' equilibria with immediate information aggregation, while complements characterize ``worst-possible'', delayed aggregation.
Game-theoretic applications also include settings such as crowdsourcing contests and Q\&A forums.
In an algorithmic context, where substitutes capture diminishing marginal improvement of information to an optimization problem, substitutes imply efficient approximation algorithms for a very general class of (adaptive) information acquisition problems.

In tandem with these broad applications, we examine the structure and design of informational substitutes and complements.
They have equivalent, intuitive definitions from disparate perspectives: submodularity, geometry, and information theory.
We also consider the design of scoring rules or optimization problems so as to encourage substitutability or complementarity, with positive and negative results.
Taken as a whole, the results give some evidence that, in parallel with substitutable items, informational substitutes play a natural conceptual and formal role in game theory and algorithms.
\end{abstract}

%

\clearpage

\tableofcontents

\break

\pagenumbering{arabic}
\section{Introduction and Overview} \label{sec:intro}
\subsection{Motivation and challenge}
An agent living in an uncertain world wishes to make some decision (whether to bring an umbrella on her commute, how to design her company's website, \dots).  She can improve her expected utility by obtaining pieces of information about the world prior to acting (a weatherman's forecast or a barometer reading, market research or automated A/B testing, \dots).
This naturally leads her to assign \emph{value} to different pieces of information and combinations thereof.
The value of information arises as the expected improvement it imparts to her optimization problem.

We would like to generally understand, predict, or design algorithms to guide such agents in acquiring and using information.
Consider the analogous case where the agent has value for \emph{items} or goods, represented by a valuation function over subsets of items.
A set of items are substitutes if, intuitively, each's value decreases given more of the others; they are complements if it increases.
Here, we have rich game-theoretic and algorithmic theories leveraging the structure of substitutes and complements (S\&C).
For instance, in many settings, foundational work shows that substitutability captures positive results for existence of market equilibria, while complements capture negative results~\citep{kelso1982job,roth1984stability,gul1999walrasian,hatfield2005matching,ostrovsky2008stability}.
When substitutes are captured by submodular valuation functions~\citep{lehmann2001combinatorial}, algorithmic results show how to efficiently optimize (or approximately optimize) subject to constraints imposed by the environment (\emph{e.g.} \citet{calinescu2011maximizing}).
For example, an agent wishing to select from a set of costly items with a budget constraint has a $(1-\frac{1}{e})$-approximation algorithm if her valuation function is submodular~\citep{sviridenko2004note}.

Can we obtain similar structural and algorithmic results for information?
Here, a piece of information is modeled as a \emph{signal} or random variable that is correlated in some way with the state of the world that the agent cares about (whether it will rain, how profitable are different website designs, \dots).
Intuitively, one might often expect information to satisfy substitutable or complementary structure.
For instance, a barometer reading and an observation of whether the sky is cloudy both yield valuable information about whether it will rain to an umbrella-toting commuter; but these are substitutable observations for our commuter in that each is probably worth less once one has observed the other.
On the other hand, the dew point and the temperature tend to be complementary observations for our commuter:
Rain may be only somewhat correlated with dew point and only somewhat correlated with temperature, but is highly correlated with cases where temperature and dew point are close (\emph{i.e.} the relative humidity is high).

Despite this appealing intuition, there are significant challenges to overcome in defining informational S\&C.
Pieces of information, unlike items, may have complex probabilistic structure and relationships.
But on the other hand, this structure alone cannot capture the value of that information, which (again unlike items) seemingly must arise from the context in which it is used.
Next, even given a measure of value, it is unclear how to formalize an intuition such as ``diminishing marginal value''.
Finally, it remains to demonstrate that the definitions are tractable and have game-theoretic and/or algorithmic applications.
These challenges seem to have prevented a successful theory of informational S\&C thus far.

\vfill
\break

\subsection{This paper: summary and contributions}
This paper has four components.

\vspace{1em}
\noindent$1$.~~
       We propose a definition of informational substitutes and complements (S\&C).
       Beginning from the very general notion of \emph{value of information} in the context of any specific decision or optimization problem, we define S\&C in terms of diminishing (increasing) marginal value for that problem.
       This requires a definition of ``marginal unit'' of information.
       We consider a hierarchy of three kinds of marginal information: learning another signal, learning some deterministic function of another signal, and learning some randomized function (``garbling'') of another signal.
       These correspond respectively to \emph{weak}, \emph{moderate}, and \emph{strong} versions of the definitions, formalized by strengthenings of \emph{submodularity} and \emph{supermoduarity}.

       We also investigate some useful tools and equivalent definitons.
       From an information-theoretic perspective, substitutes can be defined as signals that reveal a diminishing amount of information about a particular event, where the measure of information is some generalized entropy function.
       From a geometric perspective, substitutes can be defined using a measure of distance, namely some Bregman divergence, on the space of \emph{beliefs} about the event; signals are substitutes if the average change in one's belief about the event is dimininishing in the amount of information already known.

\vspace{1em}
\noindent$2$.~~
       We give game-theoretic applications of these definitions, primarily on information aggregation in prediction markets.
       When strategic agents have heterogeneous, valuable information, we would like to understand when and how their information is revealed and aggregated in an equilibrium of strategic play.
       Prediction markets, which are toy models of financial markets, are possibly the simplest setting capturing the essence of this question. 
       However, although the efficient market hypothesis states that information is quickly aggregated in financial markets~\citep{fama1970efficient}, despite much research on this question in economics (\emph{e.g.} \citet{kyle1985continuous,ostrovsky2012information}) and computer science (\emph{e.g.} \citet{chen2007bluffing,dimitrov2008non,gao2013jointly}), very little was previously known about how quickly information is aggregated in markets except in very special cases.

       We address the main open question regarding strategic play in prediction markets: When and how is information aggregated?
       We show that informational substitutes imply that all equilibria are of the ``best possible'' form where information is aggregated immediately, while complements imply ``worst possible'' equilibria where aggregation is delayed as long as possible.
       Furthermore, the respective converses hold as well; \emph{e.g.}, if an information structure guarantees the ``best possible'' equilibria, then it must satisfy substitutes.
       \begin{theorem*}[Informal] In a prediction where trader's signals are strict, strong substitutes with respect to the market scoring rule, in \emph{all} equilibria, traders rush to reveal and aggregate information immediately.
       Conversely, if signals are not strong substitutes, then there are arrival orders for traders where \emph{no} equilibrium has immediate aggregation.
       \end{theorem*}
       \begin{theorem*}[Informal] In a prediction market where trader's signals are strict, strong complements with respect to the market scoring rule, in \emph{all} equilibria, traders delay revealing and aggregating information as long as possible.
       Conversely, if signals are not strong complements, then there are arrival orders where \emph{no} equilibrium exhibits full delaying.
       \end{theorem*}
       Informational S\&C thus seem as fundamental to equilibria of (informational) markets as substitutable items are in markets for goods.

       We believe that informational S\&C have the potential for broad applicability in other game-theoretic settings involving strategic information revelation, and toward this end, give some additional example applications.
       We show that S\&C characterize analogous ``rush/delay'' equilibria in some models of machine-learning or crowdsourcing contests~\citep{abernethy2011collaborative,waggoner2015market} and question-and-answer forums~\citep{jain2009designing}.

\vspace{1em}
\noindent$3$.~~
       We give algorithmic applications, focusing on the complexity of approximately-optimal information acquisition.
       Namely, we define a very broad class of problems, termed \sigsel/, in which a decision maker wishes to acquire information prior to making a decision, but has constraints on the acquisition process.
       For instance, a company wishes to purchase heterogeneous, pricey data sets subject to a budget constraint, or to place up to $k$ sensors in an environment.
       We show that substitutes imply efficient approximation algorithms in many cases such as a budget constraint; this extends to an adaptive version of the problem as well.
       We also show that the problem is hard in general and in the complements case, even when signals are independent uniform bits.
       \begin{theorem*}[Informal] For the \sigsel/ problem with e.g. cardinality or budget (``knapsack'') constraints, in the \emph{oracle model} of input: If signals are weak substitutes, then polynomial-time $1-1/e$ approximation algorithms exist; but in general, or even if signals are assumed to be complements, no algorithm can achieve nonzero approximation with subexponentially many oracle queries. 
       \end{theorem*}
       These results offer a unifying perspective on a variety of similar ``submodularity-based'' solutions in the literature~\citep{krause2005optimal,guestrin2005near,krause2009optimal,golovin2011adaptive}.

\vspace{1em}
\noindent$4$.~~
       We investigate the structure of informational S\&C.
       We give a variety of tools and insights for both identifying substitutable structure and \emph{designing} for it.
       For instance, we provide natural geometric and information-theoretic definitions of S\&C and show they are equivalent to the submodularity-based definitions.

       We address two fundamental questions: Are there (nontrivial) signals that are substitutes for \emph{every} decision problem?
       Second, given a set of signals, can we always design a decision problem for which they are substitutes?
       In the game-theoretic settings above, this corresponds to design of mechanisms for immediate aggregation, somewhat of a holy grail for prediction markets.
       In algorithmic settings, it has relevance for the design of \emph{submodular surrogates}~\citep{chen2015submodular}.
       Unfortunately, we give quite general negative answers to both questions.
       Surprisingly, more positive results arise for complements.
       We give the geometric intuition behind these results and point toward heuristics for substitutable design in practice.

\vspace{1em}
In summary, the contributions of this paper are twofold: (a) in the definitions of informational S\&C, along with a body of evidence that they are natural, tractable, and useful; and (b) in the applications, in which we resolve a major open problem on strategic information revelation as well as give a unifying and general framework for a broad algorithmic problem.
Our results on structure and design of informational S\&C points to potential for these very general definitions and results to have concrete applications.

Taken all together, we believe these results give evidence that informational S\&C, in analogy with the successful theories of substitutable goods, have a natural and useful role to play in game theory, algorithms, and in connecting the two.

\subsection{Outline}
Two sections have been placed in the appendix for convenience.
Appendix \ref{sec:related} gives a detailed survey of related work in a variety of areas.
In Appendix \ref{sec:develop}, we overview and justify our general approach to defining informational S\&C, including historical context, tradeoffs, and intuition.
We particularly focus on a comparison to the proposed definitions of \citet{borgers2013signals}, which we build on in this paper.

In Section \ref{sec:defs}, we concisely and formally define informational S\&C.
We show that prediction problems, and the modern convex analysis understanding of them, can be used to analyze general decision problems.
Leveraging these tools, we give three equivalent definitions from seemingly-disparate perspectives.

In Section \ref{sec:game-theoretic}, we present game-theoretic applications.
Primarily, we show that informational substitutes (complements) characterize best-case (worst-case) information aggregation in prediction markets.

In Section \ref{sec:algorithmic}, we present algorithmic applications.
We define \sigsel/, a class of information acquisition problems, and show that substitutes correspond to efficient approximation algorithms while there are strong hardness results in general.

In Section \ref{sec:structural}, we investigate informational S\&C themselves with an eye toward the previous applications.
We give some results on general classes of S\&C and on the design of prediction or optimization problems for which a given information structure is substitutable, with both game-theoretic and algorithmic implications.

Section \ref{sec:conclusion} summarizes and discusses future work.

\clearpage

\section{Definitions and Foundations} \label{sec:defs}
\subsection{Setting: information structure and decision problems} \label{sec:setting}
We now formally present the setting and definitions.
Motivation for the choices made and relation to prior work, particularly \citet{borgers2013signals}, are described in depth in Section \ref{sec:develop}.

\subsubsection{Model of information and decision problems}
\paragraph{Information structure.}
We take a standard Bayesian model of probabilistic information.
There is a random event $E$ of interest to the decisionmaker, \emph{e.g.} $E \in \{\text{rain},\text{no rain}\}$.
There are also $n$ ``base signals'' $A_1,\dots,A_n$, modeled as random events.
These represent potential information obtained by a decision-maker, \emph{e.g.} $A_i \in \{\text{cloudy}, \text{sunny}\}$.
An \emph{information structure} is given by $E$, $A_1,\dots,A_n$, and a prior distribution $P$ on outcomes $(e,a_1,\dots,a_n)$.
For simplicity, we assume that all $A_i$ and $E$ have a finite set of possible outcomes.

In addition to the base signals, there will be other signals that intuitively represent combinations of base signals.
Formally, there is a set $\Lat$ of signals, with a generic signal usually denoted $A$ or $B$.
Any subscripted $A_i$ always refers to a base signal, while $A$ may in general be any signal in $\Lat$.
We will describe how $\Lat$ is generated from $A_1,\dots,A_n$ momentarily, in Section \ref{sec:lattices}.

We will use lower-case $p$ to refer to probability distributions on $E$, the event of interest.
The notation $p(e)$ refers to the probability that $E=e$, while $p(a_i,e) = \Pr[A_i=a_i, E=e]$, and so on.
The notation $p(e | a_i)$ refers to the probability that $E=e$ conditioned on $A_i=a_i$, obtained from the prior via a Bayesian update: $p(e | a_i) = P(e,a_i) / P(a_i)$.
We will sometimes use the shorthand notation $p_a$ to refer to the posterior distribution on $e$ conditioned on $A = a$, similarly for $p_{a,b}$ when $A=a$ and $B=b$, and so on.
We will abuse notation and write $E$ to represent a set of outcomes, so for instance we may write $e \in E$; similarly for signals.
We also sometimes write $\E_{a\sim A}\left[ \dots \right]$ for the expectation over outcomes $a$ of $A$.

\paragraph{Decision problems and value function.}
A single-agent \emph{decision problem} consists of a set of event outcomes $E$, a decision space $\D$, and a utility function $u: \D \times E \to \mathbb{R}$, where $u(d,e)$ is the utility for taking action $d$ when the event's outcome is $E=e$.
This decision problem, in the context of an information structure, will be how signals derive their value.

Specifically, given the prior $P$, the decision that maximizes expected utility is \mbox{$\arg\max_{d\in \D} \E_e u(d,e)$}.
But now suppose a Bayesian, rational agent knows $P$ and will first observe the signal $A$, then update to the posterior $p_{a}$ on $E$, and then choose a decision maximizing expected utility for this posterior belief.
In this case, her utility will be given by the following ``value'' function:
\begin{equation}
 \V^{u,P}(A) := \E_a \left[ ~ \max_{d \in \D} ~ \E_e \left[ u(d,e) \mid A=a\right] ~ \right]. \label{eqn:def-value}
\end{equation}
We will use $\bot$ to denote a null signal, so that $\V^{u,p}(\bot)$ is the expected utility for deciding based only on the prior distribution.
Where the decision problem and information structure are evident from context, we will omit the superscripts $u,P$.

Intuitively, $\V^{u,P}$ is analogous to a valuation function $v: 2^{\{1,\dots,n\}} \to \mathbb{R}$ over subsets of items.
However, inputs to $\V$ may not only represent subsets of $A_1,\dots,A_n$, but also signals that give partial information about them.

\subsubsection{Signal lattices} \label{sec:lattices}
We will consider three kinds of signal sets $\Lat$, leading to ``weak'', ``moderate'', and ``strong'' substitutes and complements.
In each case, the set of signals $\Lat$ will form a lattice.
\begin{definition} \label{def:lattice}
A \emph{lattice} $(U,\preceq)$ is a set $U$ together with a partial order $\preceq$ on it such that for all $A,B \in U$, there are a \emph{meet} $A \meet B$ and \emph{join} $A \join B$ in $U$ satisfying:
\begin{enumerate}
 \item $A \meet B \preceq A \preceq A \join B$ and $A \meet B \preceq B \preceq A \join B$; and
 \item the meet and join are the ``highest'' and ``lowest'' (respectively) elements in the order satisfying these inequalities.
\end{enumerate}
In a lattice, $\bot$ denotes the ``bottom'' element and $\top$ the ``top'' element, \emph{i.e.} $\bot \preceq A \preceq \top$ for all $A \in U$, if they exist.
\end{definition}

The following definition illustrates one very common lattice, that of subsets of a ground set.
\begin{definition} \label{def:subsets-lattice}
The \emph{subsets signal lattice} generated by $A_1,\dots,A_n$ consists of an element $A_S$ for each subset $S$ of $\{A_1,\dots,A_n\}$, where $A_S$ is the signal conveying all realizations $\{A_i = a_i : i \in S\}$.
Its partial order is $A_S \preceq A_{S'}$ if and only if $S \subseteq S'$.
Hence, its meet operation is given by set intersection and join by set union.
\end{definition}
The bottom element $\bot$ of the subsets lattice exists and is a null signal corresponding to the empty set (we will use this notation somewhat often), while the top element also exists and corresponds to observing all signals.
Also, the partial ordering $\preceq$ denotes \emph{less informative}.
These facts will continue to hold for the other two signal lattices we define.

For the other two lattices, we utilize the main idea from the classic model of information due to \citet{aumann1976agreeing}.
Let the set $\Gamma \subseteq A_1 \times \cdots \times A_n$ consist of all signal realizations $(a_1,\dots,a_n)$ in the support of the prior distribution.
Now, a \emph{partition} is a collection of subsets of $\Gamma$ such that each $\gamma \in \Gamma$ is in exactly one subset.
Each signal $A_i$ corresponds to a partition of $\Gamma$ with one subset for each outcome $a_i$, namely, the set of realizations $\gamma = (\cdots, a_i, \cdots)$.
Example \ref{ex:aumann-partition}, given after the definition of discrete signal lattice, illustrates the partition model.

As in Aumann's model, the partitions of $\Gamma$ form a lattice, each partition corresponding to a possible signal.
The partial ordering is that $A \preceq B$ if the partition of $A$ is ``coarser'' than that of $B$.
One partition is \emph{coarser} than another (which is \emph{finer}) if each element of the former is partitioned by elements of the latter.
The join of two partitions is the coarsest common refinement (the coarsest partition that is finer than each of the two), while the meet is the finest common coarsening.
Example \ref{ex:aumann-coarse}, given after the definition, illustrates coarsenings and refinements.

\begin{definition} \label{def:discrete-lattice}
The \emph{discrete signal lattice} generated by $A_1,\dots,A_n$ consists of all signals corresponding to partitions of $\Gamma$, where $\Gamma$ is the subset of $A_1 \times \cdots \times A_n$ with positive probability.
Its partial order has $A \preceq B$ if the partition associated to $A$ is coarser than that of $B$.
\end{definition}
\begin{example}[Signals modeled as partitions] \label{ex:aumann-partition}
\textit{(a)} We have two independent uniform bits $A_1$ and $A_2$.
In this case $\Gamma = \{(0,0), (0,1), (1,0), (1,1)\}$.
Here $A_1$ is modeled as the partition consisting of two elements: $\{(0,0),(0,1)\}$ and $\{(1,0),(1,1)\}$.
The first element of the partition is the set of realizations where $A_1 = 0$, while the second is the set of realizations where $A_1 = 1$.

\textit{(b)} Now modify the example so that $A_1$ and $A_2$ are perfectly correlated: with probability $0.5$, $A_1 = A_2 = 0$, and with probability $0.5$, $A_1 = A_2 = 1$.
Here, $\Gamma = \{(0,0), (1,1)\}$ and $A_1$ corresponds to the partition consisting of $\{(0,0)\}$ and $\{(1,1)\}$.

\textit{(c)} Now revisit the first case where $A_1$ and $A_2$ are independent.
Imagine an agent who observes both base signals and wishes to reveal only the XOR $A_1 \oplus A_2$ of the bits.
This new signal released by the agent, call it $A$, may also be modeled as a partition of $\Gamma$, where the elements of the partition are $\{(0,0),(1,1)\}$ and $\{(0,1),(1,0)\}$.
\end{example}
\begin{example}[The order given by coarsenings] \label{ex:aumann-coarse}
\textit{(a)} We have a single signal $A_1$ which is distributed uniformly on $\{1,2,3,4,5,6\}$.
Then $\Gamma = \{1, 2, 3, 4, 5, 6\}$ and $A_1$'s partition contains these six subsets: $\{1\},\{2\},\{3\},\{4\},\{5\},\{6\}$.

\textit{(b)} Given the above information structure, suppose that an agent holding $A_1$ will commit to releasing some deterministic function of $A_1$.
In terms of information revealed, the agent may map each realization $a_1 \in \{1,\dots,6\}$ to a different report -- this is the same as just revealing $a_1$ -- or she may map some realizations to the same report.
Suppose that she reports ``small'' whenever $a_1 \in \{1,2,3\}$ and reports ``large'' whenever $a_1 \in \{4,5,6\}$.
The information revealed by this report is captured by a binary signal $A$ corresponding to the partition with two elements: $\{1,2,3\}$ and $\{4,5,6\}$.
The partition of $A$ is coarser than that of $A_1$, so $A \preceq A_1$ on the discrete lattice.

\textit{(c)} Given the same information structure, imagine that the agent will commit to releasing ``even'' whenever $a_1 \in \{2,4,6\}$ and ``odd'' whenever $a_1 \in \{1,3,5\}$.
This corresponds to a signal $B$ whose partition has these two elements and is again coarser than that of $A_1$.

\textit{(d)} Consider the above two signals $A$ and $B$.
They are incomparable: Neither is coarser nor finer than the other.
The meet $A \meet B$ will be the null signal\footnote{Formally, this is the signal whose partition contains a single element: all of $\Gamma$.} $\bot$, intuitively because given $A$, one cannot guarantee anything about the outcome of $B$.
The join $A \join B$ intuitively corresponds to observing both signals.
Let $C = A \join B$.
The partition corresponding to $C$ has the following four elements: $\{1,3\}$, $\{2\}$, $\{4,6\}$, $\{5\}$.
These each correspond to a realization of the signal $C$; call the realizations respectively $c_1,c_2,c_3,c_4$.
Here, when for example $A =$ ``small'' and $B =$ ``even'', then $C = c_2$ and an observer of $C$ would know that $A_1 = 2$.
When $A =$ ``large'' and $B =$ ``even'', then $C = c_3$ and an observer of $C$ would know that $A_1 \in \{4,6\}$, updating to a posterior on these possibilities.
\end{example}

\vspace{1em}
For the third and strongest notion, we extend the model by, intuitively, appending randomness to the signals on the discrete lattice.
Given any signal $A$ on the discrete lattice, a ``garbling'' of $A$ can be captured by a randomized function of $A$; but this may be modeled as a \emph{deterministic} function $s(A,r)$ where $r$ is a uniform $[0,1]$ random variable\footnote{In some applications, it may be more desirable to use an infinite string of independent uniform bits.}.
This observation allows us to ``reduce'' to the deterministic case, but where each possible signal carries extra information in the form of some independent randomness.

Specifically, let $\Gamma$ be defined as above (the subset of $A_1 \times \cdots \times A_n$ with positive probability) and, for each partition $\Pi$ of $\Gamma$, let $R_{\Pi} \in [0,1]$ drawn independently from the uniform distribution.
Let $\Gamma' = \Gamma \times \mathbf{R}$ where $\mathbf{R} = \bigtimes_{\Pi} R_{\Pi}$.
Now, we proceed as before, but using $\Gamma'$.\footnote{To be precise, we should restrict to measurable subsets using the Lebesgue measure on $[0,1]$.
  We will omit this discussion for simplicity; if concerned, the reader may alternatively assume that each $R_{\Pi}$ is drawn uniformly from a massive but finite set, with some tiny $\epsilon$ approximation carried through our results.}
\begin{definition} \label{def:continuous-lattice}
The \emph{continuous signal lattice} consists of a signal corresponding to each partition of $\Gamma'$.
Its partial order has $A \preceq B$ if the partition associated to $A$ is coarser than that of $B$.
\end{definition}

\begin{example}[Modeling garblings via the continuous lattice] \label{ex:continuous-strategy}
Consider a uniformly random bit $A_1$ as the only base signal; the resulting $\Gamma$ is $\{0, 1\}$.
Now consider the garbling where, if $A_1 = 0$, then output ``happy'' with probability $q_0$ and ``sad'' otherwise; if it equals $1$, then output ``happy'' with probability $q_1$ and ``sad'' otherwise.
Call the output of the garbling $A$.
Then $A$ can be modeled as a partition of $\Gamma \times [0,1]$ with the following two subsets: $\{(0,x): 0 \leq x \leq q_0\} \cup \{(1,x): 0 \leq x \leq q_1\}$, and $\{(0,x) : q_0 < x \leq 1\} \cup \{(1,x) : q_1 < x \leq 1\}$.
Here the first realization of $A$ corresponds to the output ``happy'', while the second corresponds to output ``sad''.
To see this, note for instance that the first realization contains all the elements of $\Gamma \times [0,1]$ where $A_1 = 0$ and the randomness variable $x \leq q_0$.
So when $A_1 = 0$, assuming $x$ is drawn uniformly and independently from $[0,1]$, then the outcome of $A$ is ``happy'' with probability $q_0$.
On the continuous lattice, $A_1$ corresponds to the partition of singletons such as $\{(0,0.35142)\}$, $\{(1,0.92241)\}$, and so on.
That is, it corresponds to observing both the original binary bit as well as the random real number $x$.
Because this partition is finer than that corresponding to $A$, we have $A \preceq A_1$ on the continuous lattice.
\end{example}

The use of ``happy'' and ``sad'' for the outputs in the above example illustrates that it is not important, when considering the information conveyed by signal $A$, to consider what its realizations were \emph{named}.
All that matters is their distributions, e.g. the partitions they represent.

\begin{example}[Modeling garblings, continued] \label{ex:continuous-strategy-2}
Again let $A_1$ be a uniformly random bit, and now suppose $A$ is obtained by adding to $A_1$ independent Gaussian noise with mean $0$ and variance $1$.
In this case, intuitively, each outcome of $A$ (say $A=0.3$) represents two possibilities (such as $A_1 = 0$ and the Gaussian is $0.3$, or $A_1 = 1$ and the Gaussian is $-0.7$).
$A$ can be modeled as a partition of $\Gamma \times [0,1]$ where $x \in [0,1]$ is interpreted as the quantile of the outcome of the Gaussian.
Each member of the partition has two elements.
These can be written $(0,x_0)$ and $(1,x_1)$ with $x_0 = \Phi(A)$ and $x_1 = \Phi(A-1)$, using the standard normal CDF $\Phi$.
Given the realization $A=0$, the posterior distribution on $A_1$ is given by a Bayesian update depending on the probability density of the Gaussian at $0$ and at $-1$.
\end{example}

\subsection{The definitions of substitutes and complements} \label{sec:defs-subs-comps}
We utilize a common notion of diminishing and increasing marginal value.
For example, the idea of submodularity is that a lattice element $B$'s marginal contribution to $A$ should be smaller that to some $A' \preceq A$.
\begin{definition} \label{def:sub-super-modular}
\defsubsupermodular
\end{definition}
\begin{definition}[Weak Informational S\&C] \label{def:subs-comps}
\defsubscompsweak
\end{definition}
For moderate and strong substitutes, we will use a strengthening of submodularity by requiring diminishing marginal value with respect to, respectively, all deterministic and randomized functions of a signal.\footnote{Previous versions of this paper used submodularity on the discrete and continuous signal lattices, which is a more restrictive definition. We only need this weaker definition for results in this paper, but in general both are interesting and the right choice may be context-dependent; or there could be other interesting variations.}
\begin{definition}[Moderate and Strong Informational S\&C] \label{def:mod-strong-subs-comps}
\defsubscompsmodstrong
\end{definition}
Verbally, S\&C capture that the more pieces of information one has, the less valuable (respectively, more valuable) $B$ becomes.
The levels of weak, moderate, and strong capture the senses in which ``pieces of information'' is interpreted.
Weak substitutes satisfy diminishing marginal value whenever a whole signal is added to a subset of signals.
However, they do not give guarantees about marginal value with respect to partial information about signals.
Moderate and strong substitutes respectively satisfy diminishing marginal value when deterministic (randomized) partial information about a signal is revealed.

\begin{observation}
Strong substitutes imply moderate substitutes, which imply weak substitutes.
The same holds for complements.
\end{observation}
\begin{proof}
The respective lattices are supersets, i.e. the continuous signal lattice is a superset of the discrete lattice which is a superset of the subsets lattice, and the partial orderings agree.
So each substitutes definition requires $\V^{u,P}$ to satisfy a set of inequalities at various points, and in going from weak to moderate to strong substitutes, we simply increase the set of required inequalities that signals must satisfy.
\end{proof}

\begin{example}[Substitutes] \label{ex:one-bit-subs}
The event $E$ is a uniformly random bit and the two signals $A_1 = E$ and $A_2 = E$.
That is, both signals are always equal to $E$.
The decision problem is to predict the outcome of $E$ by deciding either $0$ or $1$, with a payoff of $1$ for correctness and $0$ otherwise.
In this case, one can immediately see that $A_1$ and $A_2$ are \emph{e.g.} weak substitutes, as a second signal never gives marginal benefit over the first.
\end{example}

\begin{example}[Complements] \label{ex:one-bit-comps}
The event $E$ and decision problem are the same as in Example \ref{ex:one-bit-subs}, but this time $A_1$ and $A_2$ are uniformly random bits with $E = A_1 \oplus A_2$, the XOR of $A_1$ and $A_2$.
In this case, $A_1$ and $A_2$ are immediately seen to be \emph{e.g.} weak complements, as a first signal never gives marginal benefit over the prior.
\end{example}

\begin{example}[Weak vs moderate] \label{ex:weak-but-not-moderate}
Here is an example of weak substitutes that are not moderate substitutes.
Intuitively, we will pair the previous two examples.
The event $E$ consists of a pair $(E_b,E_c)$ of independent uniformly random bits.
The decision problem is to predict both components of $E$, getting one point for each correct answer.
Let the random variable $B_1 = E_b$ and $B_2 = E_b$.
Let the random variables $C_1$ and $C_2$ be uniformly random bits such that $E_c = C_1 \oplus C_2$.

Now, consider the signals $A_1 = (B_1,C_1)$ and $A_2 = (B_2,C_2)$.
Intuitively, the first component of each signal completely determines $E_b$, while the second component gives no information about $E_c$ until combined with the other signal.
Hence these signals intuitively have both substitutable and complementary internal structure.
Consider the subsets lattice $\{\emptyset, \{A_1\}, \{A_2\}, \{A_1,A_2\}\}$.
If we modify the decision problem such that predicting the first component of $E$ is worth $1+\epsilon$ points, then these signals are weak substitutes: Each alone is worth $1+\epsilon$ points, while together they are worth $2+\epsilon$ points.
On the other hand, if we modify the decision problem such that the second component of $E$ is worth $1+\epsilon$ points, then these signals become weak complements for analogous reasons.

On the other hand, these signals are neither moderate substitutes nor moderate complements.
One way to see this is to consider ``coarsening'' $A_1$ into the signal $B_1$; this has diminishing marginal value when added to $A_2$.
However, we could also coarsen $A_1$ into the signal $C_1$, which has increasing marginal value when added to $A_2$.
\end{example}

\subsection{Scoring rules and a revelation principle} \label{sec:revelation}
We now introduce proper scoring rules and prove a useful ``revelation principle''.

A \emph{scoring rule} for an event $E$ is a function $S: \Delta_E \times E \to \mathbb{R}$, so that $S(\hat{q},e)$ is the score assigned to a prediction (probability distribution) $\hat{q}$ when the true outcome realized is $E=e$.
Define the useful notation $S(\hat{q};q) = \E_{e\sim q} S(\hat{q},e)$ for the expected score under true belief $q$ for reporting $\hat{q}$ to the scoring rule.

The scoring rule is \emph{(strictly) proper} if for all $E,q$, setting $\hat{q} = q$ (uniquely) maximizes the expected score $S(\hat{q};q)$.
In other words, if $E$ is distributed according to $q$, then truthfully reporting $q$ to the scoring rule (uniquely) maximizes expected score.

A fundamental characterization of scoring rules is as follows:
\begin{fact}[\cite{mccarthy1956measures,savage1971elicitation,gneiting2007strictly}] \label{fact:scoring-char}
For every (strictly) proper scoring rule $S$, there exists a (strictly) convex function $G: \Delta_E \to \R$ with (1) $G(q) = S(q;q)$ and (2)
 \[ S(\hat{q},e) = G(\hat{q}) + \langle G'(\hat{q}), \delta_e - \hat{q} \rangle  \]
where $G'(\hat{q})$ is a subgradient of $G$ at $\hat{q}$ and $\delta_e$ is the probability distribution on $E$ putting probability $1$ on $e$ and $0$ elsewhere.

Furthermore, for every (strictly) convex function $G: \Delta_E \to \R$, there exists a (strictly) proper scoring rule $S$ such that (1) and (2) hold.
\end{fact}
\begin{proof}
Given any (strictly) convex $G$, we first check that the induced $S$ is (strictly) proper.
Select a subgradient $G'(p)$ at each point $p$.
The expected score for reporting $\hat{q}$ when $E$ is distributed according to $q$ is
\begin{align*}
 S(\hat{q};q)
  &= \E_{e\sim q} S(\hat{q},e)  \\
  &= G(\hat{q}) + \langle G'(\hat{q}), q - \hat{q}\rangle  \\
  &\leq G(q)  & \text{by convexity of $G$}  \\
  &= S(q;q) .
\end{align*}
Note that the inequality follows simply because, for any convex $G$, if we take the linear approximation at some point $\hat{q}$ and evaluate it at a different point $q$, this lies below $G(q)$.
Furthermore, if $G$ is strictly convex, then this inequality is strict, implying strict properness.

Now, given a (strictly) proper $S$, we show that it has the stated form.
Define $G(q) = S(q;q)$.
Note that $S(\hat{q};q) = \E_{e\sim q} S(\hat{q},e)$ is a linear function of $q$.
By properness, each $G(q) = S(q;q) = \max_{\hat{q}} S(\hat{q};q)$.
Since $G(q)$ is a pointwise maximum over a set of linear functions of $q$, $G$ is convex.
If $S$ was strictly proper, then $G(q)$ was the unique maximum at every point, implying that $G$ is strictly convex.

Now we claim that $S(q; \cdot)$ is a subtangent of $G$ at $q$: it is linear, equal to $G$ at $q$, and everywhere below $G$ by definition of $G$.
So in particular $S(q,e) = S(q;\delta_e) = G(q) + \langle G'(\hat{q}), \delta_e - q\rangle$, as promised.
\end{proof}

\begin{figure}[ht]
\centering
\includegraphics[width=0.8\textwidth]{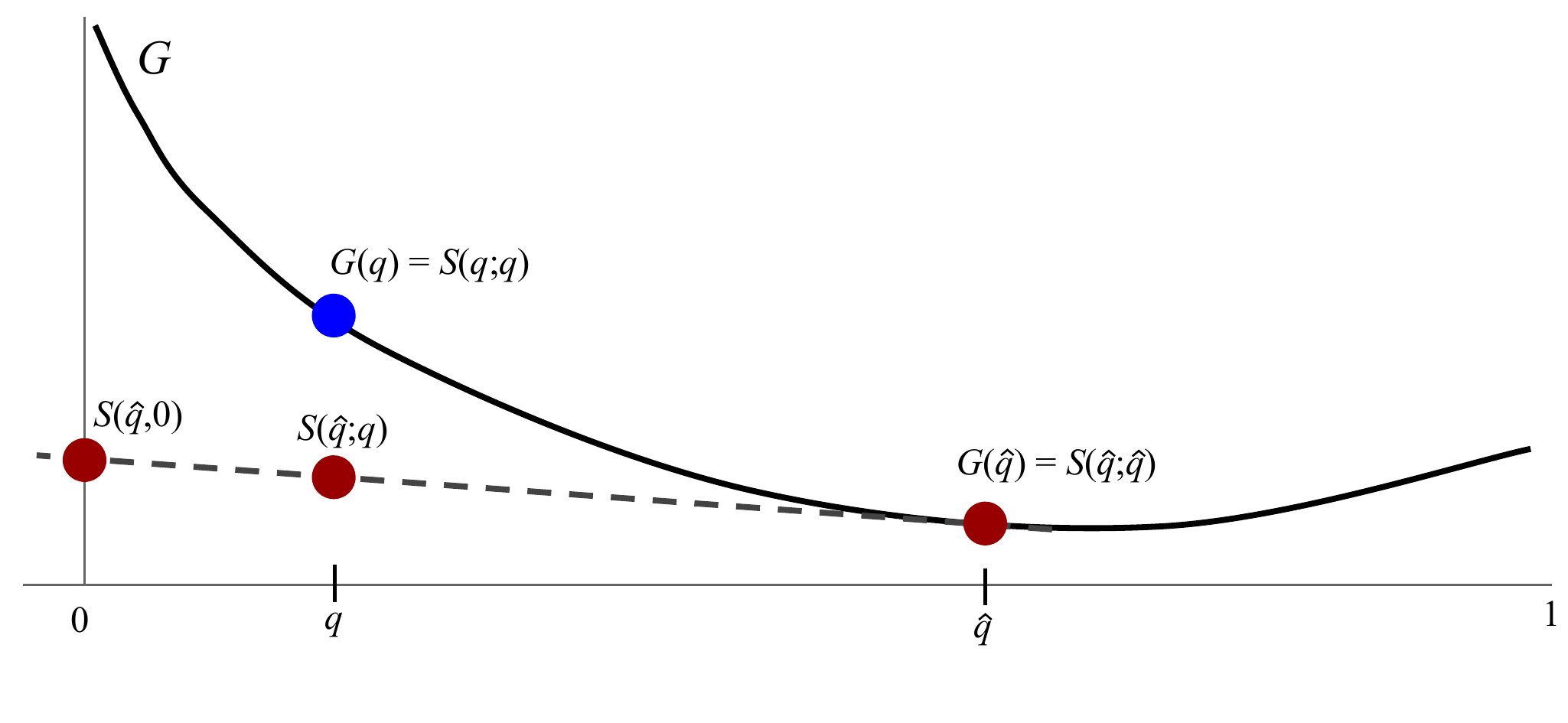}
\caption{Illustration of the connection between a scoring rule $S$ and $G$, its associated convex expected score function, for a binary event $E$ (Fact \ref{fact:scoring-char}).
  The $x$-axis is the probability that $E=1$.
  $S(\hat{q},e)$ is the score obtained for predicting $\hat{q}$ when $E=e$, while $S(\hat{q};q)$ is the expected score for predicting $\hat{q}$ when the believed distribution over $E$ is $q$.
  By the characterization, $S(\hat{q};q) = G(\hat{q}) + \langle G'(\hat{q}), q-\hat{q}\rangle$, which is pictured by taking the linear approximation to $G$ at $\hat{q}$ and evaluating it at $q$.
  Convexity of $G$ implies that this linear approximation is always below $G$, hence reporting truthfully is optimal.
  We return to this picture when discussing the relationship to Bregman divergences.}
\label{fig:scoringrule}
\end{figure}

\begin{example} \label{ex:def-log-scoring-rule}
The $\log$ scoring rule is $S(p,e) = \log p(e)$, \emph{i.e.} the logarithm (usually base $2$) of the probability assigned to the realized event.
The expected score function is $\sum_e p(e) \log p(e) = -H(p)$, where $H$ is the Shannon entropy function.
\end{example}

Notice that a scoring rule is a special case of a decision problem: The utility function is the scoring rule $S$, $E$ is the event picked by nature, and the decision space $\D = \Delta_E$.
We now show that in a sense, scoring rules capture \emph{all} decision problems.
This is not surprising or difficult, and may have been observed prior to this work; but we formalize it because it captures a very nice and useful intuition.
\begin{theorem}[Revelation principle] \label{thm:revelation}
\theoremrevelation
\end{theorem}
\begin{proof}
The idea of the proof, as suggested by the name, is simply for the agent to report her belief $q$ about $E$ to the scoring rule and for the scoring rule to simulate the optimal decision for this belief, paying the agent according to the utility derived from that decision.
For a given distribution (``belief'') $q$ on $E$, let $d^*_q$ be the optimal decision, \emph{i.e.} $d^*_q = \arg\max_{d\in\D} \E_{e\sim q} u(d,e)$.
Now, given $u,\D,E$, let
 \[ S(\hat{q},e) = u(d^*_{\hat{q}}, e).  \]
First let us show properness, \emph{i.e.} $S(\hat{q};q) \leq S(q;q)$.
We have
\begin{align*}
 S(\hat{q};q)
  &= \E_{e\sim q} S(\hat{q},e)  \\
  &= \E_{e\sim q} u(d^*_{\hat{q}}, e)  \\
  &\leq \E_{e\sim q} u(d^*_q, e)  \\
  &= \E_{e\sim q} S(q,e)  \\
  &= S(q;q)
\end{align*}
using the definition of $d^*_q$.

Now let us check equivalence to the original problem. Let $q_a$ be the distribution on $E$ conditioned on $A=a$.
We have
\begin{align*}
 \V^{u,P}(A) &= \E_a \max_{d \in \D} \E_{e\sim q_a} u(d,e)  \\
  &= \E_a \E_{e\sim q_a} u(d^*_{q_a}, e)  \\
  &= \E_a \E_{e\sim q_a} S(q_a, e)  \\
  &= \E_a \max_{\hat{q}} \E_{e \sim q_a} S(\hat{q}, e)  & \text{by properness}  \\
  &= \V^{S,P}(A) .
\end{align*}
\end{proof}
This reduction is not necessarily computationally efficient, because the input $\hat{q}$ to the scoring rule is a probability distribution over $E$ which may have a large number of outcomes.
We note two positives, however.
First, the reduction does not necessarily need to be computationally efficient to be useful for proofs and analysis.
Second, in any case where it seems reasonable to assume that the agent can solve her decision problem, which involves an expectation over possible outcomes of $E$, it seems reasonable to suppose that she can efficiently represent or query her beliefs.
In this case we may often expect a computationally efficient reduction and construction of $S$.
This is a direction for future work.

The revelation principle (Theorem \ref{thm:revelation}) and scoring rule characterization (Fact \ref{fact:scoring-char}) together imply the following extremely useful fact about general decision problems.
We do not claim originality for it; the idea can be found in \citet{savage1971elicitation} and similar ideas or statements are present in \emph{e.g.} \citet{frongillo2014general} and \citet{babaioff2012optimal}.
But it is worth emphasizing because we will put it to extensive use in this paper.
\begin{corollary} \label{cor:decision-captured-by-G}
\cordecisioncapturedbyG
\end{corollary}

As an example of usefulness, we provide a concise proof of the following classic theorem.
\begin{fact}[More information always helps] \label{fact:more-info-helps}
In any decision problem, for any signals $A,B$, $\V(A \join B) \geq \V(A)$.
In other words, more information always improves the expected utility of a decision problem.
In other words, $\V$ is a monotone increasing function on the signal lattices.
\end{fact}
\begin{proof}
Recall that we are using the notation $p_{a_1}$ for the distribution on $E$ conditioned on $A_1=a_1$, and so on.
In particular, $p_{a_1}$ is a vector, \emph{i.e.} $p_{a_1} = (p(e_1 | A_1=a_1), \dots)$.
By the revelation principle, for some convex $G$ we have $\V(A_1) = \E_{a_1} G(p_{a_1})$, and
\begin{align*}
 V(A_1 \join A_2)
  &= \E_{a_1} \left[ \sum_{a_2} p(a_2|a_1) G(p_{a_1a_2}) \right]  \\
  &\geq \E_{a_1} G\left(\sum_{a_2} p(a_2|a_1) p_{a_1a_2}\right)  & \text{by Jensen's inequality}  \\
  &= \E_{a_1} G(p_{a_1})  \\
  &= \V(A_1) .
\end{align*}
To obtain the last equality: Each term in the sum consists of the scalar $p(a_2 | a_1)$ multiplied by the vector $p_{a_1a_2}$, and for each coordinate $e$ of the vector, we have $p(a_2|a_1)p(e|a_1,a_2) = p(e,a_2|a_1)$.
Then $\sum_{a_2} p(e,a_2|a_1) = p(e|a_1)$.
\end{proof}

\subsection{Characterizations} \label{sec:characterizations}
In this section, we show how the substitutes and complements conditions can be phrased using the convexity connection just derived.
We will leverage this structure to identify characterizations or alternative definitions of substitutes and complements.
For brevity, we will focus on substitutes, but in all cases the extension to complements is immediate.

From Corollary \ref{cor:decision-captured-by-G}, we also get immediately the following characterization:
\begin{definition}[Substitutes via convex functions] \label{def:subs-G}
\defsubsG
\end{definition}
To be clear, our use of parentheses here means that we are giving three definitions: one for weak substitutes, where the condition must hold for all $A'$ on the subsets lattice; and analogously for moderate substitutes with the discrete lattice and strong substitutes with the continuous lattice.

We view Definition \ref{def:subs-G} mostly as a tool, although it may convey some intuition on its own as well.
Definition \ref{def:subs-G} will be pictured in Figure \ref{fig:divergencesubs} along with the final characterization.

\subsubsection{Generalized entropies} \label{sec:defs-entropy}
Here, we seek an alternative interpretation of the definitions of S\&C in terms of information and uncertainty.
To this end, for any decision problem, consider the convex expected score function $G$ and define $h = -G$.
Then $h$ is concave, and we interpret $h$ as a \emph{generalized entropy} or measure of information.
The justification for this is as follows: Define the notation $h(E | A) = \E_{a\sim A} h(p_a)$, where $p_a$ is the distribution on $E$ conditioned on $A=a$.
Then concavity of $h$ implies via Jensen's inequality that for all $E,A$, we have $h(E) \geq h(E|A)$.
In other words, more information always decreases uncertainty/entropy.

We propose that this is the critical axiom a generalized entropy must satisfy: If more information always decreases $h$, then in a sense it measures uncertainty, while if more information sometimes increases $h$, then it should not be considered a measure of uncertainty.
However, admittedly, the appeal of this definition may increase by adding additional axioms as are common in the literature, such as maximization at the uniform distribution and value zero at degenerate distributions.
Another very intriguing axiom would be a relaxation of the ``chain rule'' in either direction: $h(E | A)$ is restricted to be either greater than or less than $h(E,A) - h(A)$.
(Note that the chain rule itself, $h(E|A) = h(E,A) - h(A)$, along with concavity, uniquely characterizes Shannon entropy.)
Such axioms may have interesting consequences for informational S\&C.
Examining the structure of S\&C under such axioms represents an intriguing direction for future work.

\begin{figure}[ht]
\centering
\includegraphics[width=0.8\textwidth]{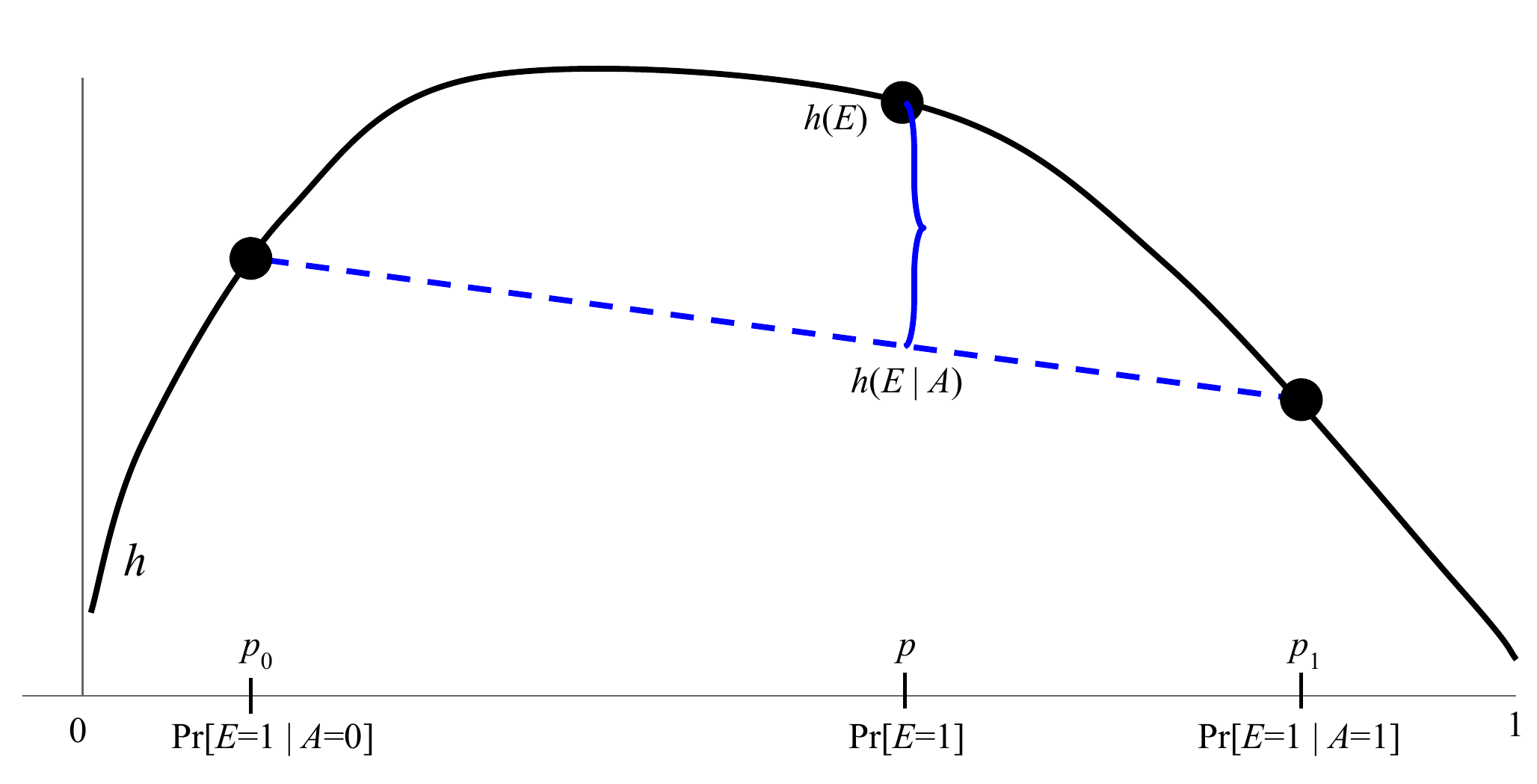}
\caption{Illustration of marginal improvement of signal $A$ over the prior, $\V(A) - \V(\bot)$, via the \emph{generalized entropy} definition used to characterize S\&C in Definition \ref{def:subs-entropy}.
  Here, the generalized entropy function $h$ captures a measure of uncertainty in a distribution over the binary event $E$.
  The marginal value of $A$ is $\V(A) - \V(\bot) = h(E) - h(E \mid A)$, the expected amount of information revealed about $E$ by $A$ (illustrated by the curly brace).}
\label{fig:entropysubs}
\end{figure}

Under this interpretation, Definition \ref{def:subs-G} can be restated:
\begin{definition}[Substitutes via generalized entropies] \label{def:subs-entropy}
\defsubsentropy
\end{definition}
Intuitively, Definition \ref{def:subs-entropy} says this:
Consider the expected amount of information about $E$ that is revealed upon learning $B$, given that some information will already be known.
Use the generalized entropy $h$ to measure this information gain.
Then substitutes imply that, the more information one has, the less information $B$ reveals.
On the other hand, complements imply that, the more information one has, the \emph{more} information $B$ reveals.

\begin{example} \label{ex:def-subs-entropy}
Revisiting Example \ref{ex:one-bit-subs}, where $E$ was a uniform bit and $A_1=A_2=E$, imagine predicting $E$ against the $\log$ scoring rule.
Our previous observations imply that here the generalized entropy function is Shannon entropy $H(q) = \sum_e q(e) \log \frac{1}{q(e)}$.
We have $H(p) = 1$ and $H(E|A_1) = H(E|A_2) = H(E|A_1,A_2) = 0$, which already shows that $A_1$ and $A_2$ are weak substitutes.

If we instead revisit Example \ref{ex:one-bit-comps}, where $A_1 \oplus A_2 = E$ with $A_1,A_2$ uniformly random bits, and again consider predicting $E$ according to the $\log$ scoring rule, then we see that $H(E) = H(E|A_1) = H(E|A_2) = 1$, while $H(E|A_1,A_2) = 0$, already proving that $A_1$ and $A_2$ are weak complements.
\end{example}

\subsubsection{Bregman divergences} \label{sec:defs-bregman}
Given a convex function $G$, the \emph{Bregman divergence} of $G$ is defined as
  \[ D_G(p,q) = G(p) - \left(G(q) + \langle G'(q), p-q\rangle\right).  \]
In other words, it is the difference between $G(p)$ and the linear approximation of $G$ at $q$, evaluated at $p$.
(See Figure \ref{fig:divergencesubs}.)
Another interpretation is to consider the proper scoring rule $S$ associated with $G$, by Fact \ref{fact:scoring-char}, and note that $D_G(p,q) = S(p;p) - S(q;p)$, the difference in expected score when reporting one's true belief $p$ versus lying and reporting $q$.
The defining property of a convex function is that this quantity is always nonnegative.
This can be observed geometrically in Figure \ref{fig:scoringrule} as well; there $D_G(q,\hat{q}) = G(q) - S(\hat{q};q)$.

This notion is useful to us because, it turns out, all marginal values of information can be exactly characterized as Bregman divergences between beliefs.
\begin{lemma} \label{lemma:marginal-contribution-bregman}
\lemmamarginalcontributionbregman
\end{lemma}
\begin{proof}
\begin{align*}
 \V(A \join B) - \V(A)
  &= \E_{a,b} G(p_{ab}) ~ - ~ \E_a G(p_a)  \\
  &= \E_{a,b} \left(G(p_{ab}) - G(p_a)\right)  \\
  &= \E_{a,b} \left(D_G(p_{ab}, p_a) - \langle G'(p_a), p_{ab} - p_a \rangle \right)  \\
  &= \E_{a,b} D_G(p_{ab}, p_a) ~ + ~ \E_a \langle G'(p_a), \sum_b p(b|a) \left(p_{ab} - p_a \right) \rangle  \\
  &= \E_{a,b} D_G(p_{ab}, p_a)
\end{align*}
because $\sum_b p(e|a,b) p(b|a) = p(e|a)$, so $\sum_b p(b|a) p_{ab} = p_a$.
\end{proof}

\begin{figure}[ht]
\centering
\includegraphics[width=0.8\textwidth]{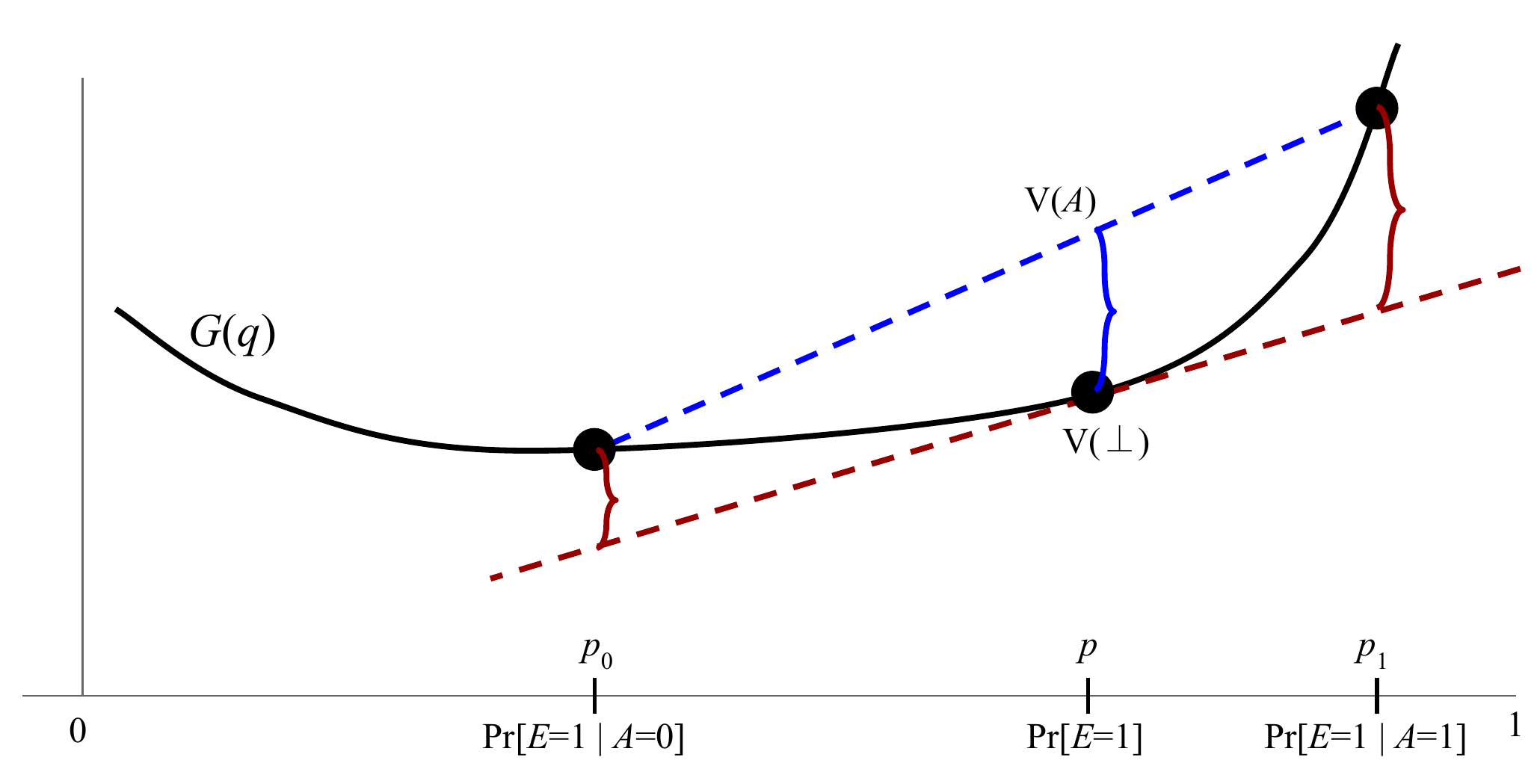}
\caption{Illustration of marginal improvement of signal $A$ over the prior, $\V(A) - \V(\bot)$, via two equivalent definitions used to characterize S\&C.
  Definition \ref{def:subs-G} is that $\V(A) - \V(\bot) = \E_a G(p_a) - G(p)$.
  Here, the blue curly brace measures the distance between $\E_a G(p_a)$ and $G(p)$.
  Definition \ref{def:subs-bregman} is that $\V(A) - \V(\bot) = \E_a D_G(p_a,p)$ where $D_G(p_a,p)$ is the Bregman divergence.
  The Bregman divergences $D_G(p_0,p)$ and $D_G(p_1,p)$ are measured by the two red curly braces.
  Another way of stating the equivalence of (1) and (2) is that the average size of the red braces is equal to the size of the blue brace (where the average is weighted by the probabilities of $A=0,1$).}
\label{fig:divergencesubs}
\end{figure}

\begin{definition}[Substitutes via divergences] \label{def:subs-bregman}
\defsubsbregman
\end{definition}
This can be interpreted as a characterization of S\&C where $D_G$ serves as a \emph{distance measure} of sorts (although it is not in general a distance metric).
The characterization says that, if we look at how ``far'' the agent's beliefs move upon learning $B$, on average, then for substitutes this distance is decreasing in how much other information is available to the agent.
But for complements, the more information the agent already has, the \emph{farther} she expects her beliefs to move on average upon learning $B$.

\begin{example} \label{ex:def-subs-divergence}
For the $\log$ scoring rule, $D_G(p,q)$ is exactly the \emph{KL-divergence} or \emph{relative entropy} $KL(p,q)$ between distributions on $p$ and $q$.
If we recall Example \ref{ex:one-bit-subs}, in which $E$ was a random bit and $A_1 = A_2 = E$, we can consider the decision problem of prediction $E$ against the $\log$ scoring rule.
In this case, the prior $p = \left(\frac{1}{2},\frac{1}{2}\right)$, while the posteriors $p_{a_1=0} = (1,0)$, $p_{a_1=1} = (0,1)$, and the same for $A_2$.
Hence $\E_{a_1} KL(p_{a_1},p) = 1$.
But the posteriors conditioned on both signals are the same, \emph{e.g.} $p_{a_0=a_2=0} = (1,0) = p_{a_0=0}$.
Hence $\E_{a_1,a_2} KL(p_{a_1,a_2},p_{a_1}) = 0$.

This already shows that $A_1$ and $A_2$ are weak substitutes.
And in fact, if $A_1 = A_2$, then this argument extends to show that $A_1$ and $A_2$ are substitutes in \emph{any} decision problem (as they should be), because given $A_1$, an update on $A_2$ moves the posterior belief a distance $0$.
\end{example}

\clearpage

\section{Game-Theoretic Applications} \label{sec:game-theoretic}
\subsection{Prediction markets} \label{sec:markets}
A prediction market is modeled as a Bayesian extensive-form game.
The market's setting is specified by a strictly proper scoring rule $S$ and an information structure with prior $P$ and event $E$, and set of signals $A_1,\dots,A_m$.
We assume these signals are ``nontrivial'' in that, given all signals but $A_i$, the distribution of $E$ changes conditioned on $A_i$.

An instantiation of the market is specified by a set of $n$ traders, each trader $i$ observing some subset of the signals (call the resulting signal $B_i$), and an order of trading $i_1,\dots,i_T$, where at each time step $t=1,\dots,T$, it is the turn of agent $i_t \in \{1,\dots,n\}$ to trade.
We assume that no trader participates twice in a row (if they do, it is without loss to delete one of these trading opportunities).

The market proceeds as follows.
First, each trader $i$ simultaneously and privately observes $B_i$, updating to a posterior belief.
Then the market sets the initial prediction $p^{(0)} \in \Delta_E$, which we assume to be the prior distribution $p$ on $E$.
We will also refer to a market prediction as the market prices.
Then, for each $t=1,\dots,T$, trader $i_t$ arrives, observes the current market prediction $p^{(t-1)}$, and may update it to (``report'') any $p^{(t)} \in \Delta_E$.

After the last trade step $T$, the true outcome $e$ of $E$ is observed and each trader $i$ receives payoff $\sum_{t: i=i_t} S(p^{(t)}, e) - S(p^{(t-1)}, e)$.
Thus, at each time $t$, trader $i_t$ is paid according to the scoring rule applied to $p^{(t)}$, but must pay the previous trader according to the scoring rule applied to $p^{(t-1)}$.
The total payment made by market ``telescopes'' into $S(p^{(T)},e) - S(p^{(0)},e)$.

At any given time step $t$, trader $i_t$ is said to be \emph{reporting truthfully} if she moves the market prediction to her current posterior belief on $E$.
In other words, she makes the myopically optimal trade.

The natural solution concept for Bayesian games is that they be in \emph{Bayes-Nash equilibrium}, where for every player, her (randomized) strategy --- specifying how to trade at each time step as a function of her signal and all past history of play --- maximizes expected utility given the prior and others' strategies.

Because this is a broad class of equilibria and can in general include undesirable equilibria involving ``non-credible threats'', it is often of interest in extensive-form games to consider the refinement of \emph{perfect Bayesian} equilibrium. Here, at each time step and for each past history, a player's strategy is required to maximize expected utility given her beliefs at that time and the strategies of the other players. (Note the difference to Bayes-Nash equilibrium in which this optimality is only required \emph{a priori} rather than for every time step.) Here, at any time step and history of play, players' beliefs are required to be consistent with Bayesian updating wherever possible. (It may be that one player deviates to an action not in the support of her strategy; in this case other players may have arbitrary beliefs about the deviator's signal.)

To be clear, every perfect Bayesian equilibrium is also a Bayes-Nash equilibrium.
Hence, we note that an existence result is strongest if it guarantees existence of perfect Bayesian equilibrium.
Meanwhile, a uniqueness or nonexistence result is strongest if it refers to Bayes-Nash equilibrium.

\paragraph{Distinguishability criterion.}
For most of our results, we will need a condition on signals equivalent or similar to those used in prior works \citep{chen2010gaming,ostrovsky2012information,gao2013jointly} in order to ensure that traders can correctly interpret others' reports.
Formally, we say that signals are \emph{distinguishable} if for all subsets $S \subseteq \{1,\dots,m\}$ and realizations $\{a_i : i \in S\}$, $\{a_i' : i \in S\}$ of the signals $\{A_i : i \in S\}$ such that, for some $i \in S$, $a_i \neq a_i'$,
 \[ \Pr[e \mid a_i : i \in S] \neq \Pr[e \mid a_i' : i \in S] . \]
We believe it may be possible to relax this criterion and/or interpret such criteria within the S\&C framework, and this is a direction for future work.

\paragraph{Our notation in prediction markets.}
Note that, for the proper scoring rule $S$ with associated convex $G$, along with the prior $P$, we have the associated ``signal value'' function $\V$:
 \[ \V(A) = \E_{a,e} S(p_{a},e) = \E_{a} G(p_{a}) . \]
In other words, $\V(A)$ is the expected score for reporting the posterior distribution conditioned on the realization of $A$.

A second key point is that, to a trader whose current information is captured by a signal $A$, the set of strategies available to that trader can be captured by the space of signals $A' \preceq A$ on the continuous lattice.
This follows because any strategy is a randomized function of her information, so the outcomes of the strategy can be labeled as outcomes of a signal $A'$.

\subsubsection{Substitutes and ``all-rush''}
We now formally define an ``all-rush'' equilibrium and show that it corresponds to informational substitutes.
The naive definition would be that each trader reports truthfully at their first opportunity, or (hence) at every opportunity.
This turns out to be correct except for one subtlety.
Consider, for example, the final trader to enter the market.
Because all others have already revealed all information, this last trader will be indifferent between revealing immediately or delaying.
Similarly, consider three traders $i,j,k$ and the order of trading $i,j,i,j,k$.
If trader $i$ truthfully reports at time $1$, then trader $j$ is not strictly incentivized to report truthfully at time $2$.
She could also delay information revelation until time $4$.
\begin{definition}
An \emph{all-rush} strategy profile in a prediction market is one where, if the traders are numbered $1,2,\dots$ in order of the first trading opportunity, then each trader $i$ reports truthfully at some time prior to $i+1$'s first trading opportunity (with the final trader reporting truthfully prior to the close of the market).
\end{definition}

Before presenting the main theorem of this section, we give the following useful lemma (which is quite well known):
\begin{lemma} \label{lemma:bne-last-truthful}
In every Bayes-Nash equilibrium, every trader reports truthfully at her final trading opportunity.
\end{lemma}
\begin{proof}
Consider a time $t$ at which trader $i = i_t$ makes her final trade.
Fix all strategies and any history of trades until time $t$; then $i$'s total expected payoff from all previous time steps is fixed as well and cannot be changed by any subsequent activity.
Meanwhile, $i$'s unique utility-maximizing action at time $t$ is to report truthfully, by the strict properness of the scoring rule.
If $i$ does not take this action, then her entire strategy is not a best response: She could take the same strategy until time $t$ and modify this last report to obtain higher expected utility.
Therefore, in Bayes-Nash equilibrium, $i$ reports her posterior on her final trading opportunity.
\end{proof}

\begin{theorem} \label{thm:subs-rush-equilibrium}
\thmsubsrushequilibrium
\end{theorem}
\begin{proof}
Let the traders be numbered in order of their first trading opportunity $1,2,\dots,n$ and let $B_i$ be the signal of trader $i$.
Before diving in, we develop a key idea.
In equilibrium, we can view the market prediction $p^{(t)}$ at time $t$ as a random variable.
Then, construct a ``signal'' $C^{(t)}$ capturing the information contained in $p^{(t)}$.
This can be pictured as the information conveyed by $p^{(t)}$ to an ``outside observer'' who knows the prior distribution and the strategy profile, but does not have any private information.
Furthermore, if traders $1,\dots,k$ have participated thus far, then $C^{(t)}$ is an element of the continuous signal lattice with $C^{(t)} \preceq B_1 \join \cdots \join B_k$, because $p^{(t)}$ is a well-defined, possibly-randomized function of $B_1 \join \cdots \join B_k$.
Finally, if all participating traders have been truthful, then $C^{(t)}$ is a member of the subsets signal lattice $\Lat$, as it exactly reveals the subset of the signals held by those traders.

Now, let $t_i^*$ be $i$'s final trading opportunity prior to $i+1$'s first trading opportunity.
We prove by backward induction on $t$ that, in BNE and for any participant $i$ participating at some time $t \geq t_i$, the following holds: if $C^{(t)}$ is an element of the subsets lattice $\Lat$, then $i$ reports truthfully at $t$.
Now, suppose we have successfully proven this claim by backward induction; let us finish the proof.
The claim implies that $C^{(t)}$ really is in $\Lat$ and $i$ really is truthful at all such time steps for the following reason: $C^{(0)}$ is the null signal and is an element of $\Lat$, so trader $1$ participates truthfully, which implies that $C^{(t_2)} \in \Lat$, which implies that $2$ participates truthfully, and so on.
So in any BNE, all participants play all-rush strategies.

Now let us prove the statement.
For the base case $t=T$, the trader participating at the final time step is truthful by Lemma \ref{lemma:bne-last-truthful}.

Now for the inductive step, consider any $t = t_i$ for some $i$.
If $t$ is $i$'s final trading opportunity, then by Lemma \ref{lemma:bne-last-truthful}, in BNE $i$ reports truthfully at $t$.
If $t < t_i^*$, then there is nothing to prove.

Otherwise, let $t'$ be $i$'s next trading opportunity after $t$.
By inductive hypothesis, $i$ is truthful at time $t'$ and thereafter.
We compute $i$'s expected utility for any strategy, and show that if $i$ is not truthful at $t$, she can improve by deviating to the following strategy: Copy the previous strategy up until $t$, report truthfully at $t$, and make no subsequent updates.

At $t'$, $i$'s strategy can be described as reporting truthfully according to $C^{(t')} = C^{(t'-1)} \join B_i$.
For this trade, $i$ obtains expected profit $\V(C^{(t')}) - \V(C^{(t'-1)})$, and $i$ obtains no subsequent profit once her information is revealed.
Meanwhile, consider $i$'s strategy at time $t$, which induces some signal $C^{(t)} \preceq B_i \join C^{(t-1)}$.
For this trade, $i$ obtains expected profit at most $\V(C^{(t)}) - \E G(p^{(t-1)})$.
This follows because a trade conveying signal $C^{(t)}$ obtains at most $\V(C^{(t)})$.\footnote{Although we don't explicitly use it here, this implies that in equilibrium, every $p^{(t)} = p_{c^{(t)}}$, that is, the price at time $t$ equals the posterior distribution on $E$ conditioned on all information that has been revealed so far, including at time $t$.}

Let $U$ be $i$'s total expected utility at time $t$ and greater.
Once $i$ reports truthfully at $t'$, she expects to make no further profit in equilibrium.
So
 \[ U \leq \V\left(C^{(t)}\right) - \E G\left(p^{(t-1)}\right) + \V\left(C^{(t'-1)} \join B_i\right) - \V\left(C^{(t'-1)}\right) . \]
Now suppose $C^{(t'-1)}$ is in $\Lat$ and $i$ is not reporting truthfully at $t$.
This implies that $B_i \not\preceq C^{(t'-1)}$.
By strong, strict substitutes, $B_i$ gives higher marginal benefit to $C^{(t)}$ than to $C^{(t'-1)}$:
 \[ \V(C^{(t'-1)} \join B_i) - \V(C^{(t'-1)}) < \V(C^{(t)} \join B_i) - \V(C^{(t)}) . \]
So
 \[ U < \V(C^{(t)} \join B_i) - \E G(p^{(t-1)}) . \]
But $i$ can achieve this by deviating to being truthful at time $t$, then not participating at any subsequent times.
(This follows because if $i$ is truthful at time $t$, she reveals the signal $C^{(t-1)} \join B_i$, which is exactly the same as $C^{(t)} \join B_i$.)
This deviation does not affect $i$'s utility from any previous times, so it is a strategy with higher total expected utility.
So $i$'s only BNE strategy can be to be truthful at $t$.
\end{proof}
\begin{theorem} \label{thm:subs-nonrush}
\thmsubsnonrush
\end{theorem}
\begin{proof}
The assumptions imply that there signals on the subsets lattice $A,B$ and some $A'$ on the continuous lattice with $A' \preceq A$ and
 \[ \V(A' \join B) - \V(A') < V(A \join B) - \V(A) . \]
Then in particular, we can consider the scenario with two traders where ``Alice'' has signal $A$ (i.e. she observes the corresponding subset of signals) and Bob has $B$, with a trading order Alice-Bob-Alice.
In perfect Bayesian equilibrium (PBE), Bob must be truthful at his trading opportunity according to his beliefs even if Alice deviates from her strategy.
By distinguishability, Alice can infer his signal from this truthful report, so in any PBE, Alice is truthful and correct in predicting $p_{ab}$ at the second opportunity.
Hence the two traders' expected utilities sum to the constant amount $\V(A \join B) - \V(\bot)$, even when Alice deviates.
If Alice reports truthfully at her first opportunity (the all-rush strategy), then Bob's expected utility is $\V(A \join B) - \V(A)$.
But if Alice reports according to $A' \preceq A$, then Bob's expected utility is at most $\V(A' \join B) - \V(A')$, which by assumption of non-substitutes is strictly smaller.
This implies that Alice prefers the deviation, so truthful reporting (and thus all-rush) could not have been an equilibrium.
\end{proof}

\subsubsection{Complements and ``all-delay''}
We begin by defining an ``all-delay'' strategy profile, analogous to all-rush.
\begin{definition}
An \emph{all-delay} strategy profile in a prediction market is one where, when the traders are numbered $1,\dots,n$ in order of their final trading opportunity, each trader $i \geq 2$ reveals no information until after trader $i-1$'s final trading opportunity.
\end{definition}

\begin{theorem} \label{thm:comps-delay-equilibrium}
\thmcompsdelayequilibrium
\end{theorem}
\begin{proof}
The ideas will be substantially the same as in Theorem \ref{thm:subs-rush-equilibrium}, but the deviation argument is somewhat trickier.
In the substitutes ``rush'' case, an agent could deviate to immediate truth-telling and ignore all subsequent consequences.
Now, we will need agents to deviate to delaying all information revelation, relying on their opponents' response to ensure this becomes profitable later.
This is also the reason that we restrict to perfect Bayesian equilibrium.

The proof is by backward induction.
We show that in any perfect Bayesian equilibrium: at each time $t$, if $C^{(t)}$ is on the subsets lattice, then players play an all-delay strategy from time $t$ onward.
(As in the proof of Theorem \ref{thm:subs-rush-equilibrium}, let $C^{(t)}$ be the signal induced by the random variable $p^{(t)}$ in equilibrium.
Because in PBE strategies are well-defined in every subgame, $C^{(t)}$ is also well-defined off the equilibrium path.)
As in Theorem \ref{thm:subs-rush-equilibrium}, this will prove the statement, because $C^{(0)}$ is on the subsets lattice, corresponding to $\emptyset$, so the first trader plays all-delay at time $1$, implying that $C^{(1)}$ is on the subsets lattice, etc.

For the base case, at $t=T$, this is the trader's final opportunity, so she is truthful by Lemma \ref{lemma:bne-last-truthful}, which constitutes an all-delay strategy from $T$ onward.

Now consider any trading time $t=T$ and participant $i$ trading at time $t$.

First suppose $t$ is $i$'s final trading opportunity; then by Lemma \ref{lemma:bne-last-truthful}, she reports truthfully at this time.
By induction, traders play all-delay at all times after $t$, so this shows that they play all-delay from time $t$ onward.

If $t$ is not a trader's final trading opportunity, but is after $i-1$'s final trading opportunity, then there is nothing to prove for this time step.
So suppose trader $i$ is trading at time $t$ with $i-1$'s final opportunity coming at some $t_{i-1} > t$.
The inductive assumption implies that in any subgame starting at time $t+1$ in PBE, $i$ does not make any update until some $t' > t_{i-1}$.
It also implies that no other trader participates between $t_{i-1}$ and $t'$.
Finally, it implies that all traders participating between time $t$ and $t'$ (exclusive) report truthfully\footnote{In a trading order such as $i,j,k,j,k$, it is possible that $j$ reports something nontrivial at time $2$, then reports truthfully at time $4$.
  But $k$ does not participate at time $3$, so $j$'s multiple reports WLOG telescope into a single truthful report.}.
Let $B$ denote the join of their signals.
Then $i$'s total utility from time $t$ onward is
 \[ U = \V\left(C^{(t)}\right) - \E G\left(p^{(t-1)}\right) + \V\left(C^{(t)} \join B \join B_i\right) - \V\left(C^{(t) \join B}\right) . \]
Now suppose for contradiction that $i$ reveals some nontrivial information at time $t$, \emph{i.e.} $C^{(t)} \neq C^{(t-1)}$.
Then $i$ can deviate to revealing nothing at time $t$, reporting according to $C^{(t-1)}$, and being truthful at time $t'$.
In this case, by assumption of PBE, others continue to best-respond.
By inductive assumption, in any subgame of a PBE (which itself must be in PBE), others who participate between $t$ and $t'$ therefore continue to report truthfully, implying that $B$ is still revealed between time $t$ and $t'$.
Now, strong, strict complements imply
 \[ \V\left(C^{(t)}\right) - \V\left(C^{(t-1)}\right) < \V\left(C^{(t)} \join B\right) - \V\left(C^{(t-1)} \join B\right) . \]
So
 \[ U < \V\left(C^{(t-1)}\right) - \E G\left(p^{(t-1)}\right) + \V\left(C^{(t)} \join B \join B_i\right) - \V\left(C^{(t-1)} \join B\right) . \]
But this is $i$'s utility for the deviation above (note that $C_t \join B \join B_i = C_{t-1} \join B \join B_i$).
Since the deviation is profitable, this gives a contradiction, implying that $i$ (and all players) must play all-delay starting from time $t$ in any PBE.
\end{proof}

\begin{theorem} \label{thm:comps-nondelay}
\thmcompsnondelay
\end{theorem}
\begin{proof}
Analogous to the substitutes case (Theorem \ref{thm:subs-nonrush}).
The assumptions imply that there are subsets-lattice signals $A,B$ and continuous-lattice signal $A' \preceq A$ such that
 \[ \V(A' \join B) - \V(A') > V(A \join B) - \V(A) . \]
Then in particular, we can consider the scenario where ``Alice'' has signal $A$ and Bob has $B$, with a trading order Alice-Bob-Alice.
In PBE, Bob is truthful even if Alice deviates.
By distinguishability, Alice can infer his signal from this truthful report, so in any PBE, Alice is truthful and correct in predicting $p_{ab}$ at the second opportunity.
So utilities have the constant sum $\V(A \join B) - \V(\bot)$ even when Alice deviates.
If Alice reports nothing at her first opportunity, then Bob's expected utility is $\V(B) - \V(\bot)$.
But if Alice deviates to reporting to $A' \preceq A$, then Bob's expected utility is at most $\V(A' \join B) - \V(A')$, which is strictly smaller.
This implies that Alice prefers the deviation, so truthful reporting (and thus all-rush) could not have been an equilibrium.

Hence, in perfect Bayesian equilibrium, trader $1$ cannot play all-delay.
\end{proof}

\paragraph{Discussion.}
These results show that informational S\&C are in a sense unavoidable in the study of settings such as prediction markets.
However, the result raises many interesting questions for future work.
Two major questions are: How can we identify structures that are substitutes or complements? and How can we \emph{design} markets to encourage substitutability?

We give some initial steps toward answering these questions in Section \ref{sec:structural}.

\subsection{Other game-theoretic applications} \label{sec:game-other}
We will now examine a few game-theoretic contexts in which our results have immediate applications or implications.
Instead of developing full formal proofs and theorem statements, we focus on illustrating the intuition of how to extend our results and the conceptual lens of informational S\&C to those settings.

\subsubsection{Crowdsourcing and contests}
In the study of crowdsourcing from a theoretical perspective, the ``crowd'' is a group of agents who hold valuable information and the goal is to design mechanisms that elicit this information.
Specifically, here we are interested in ``wisdom of the crowd'' settings where the total information available to the crowd is greater than that of the most-informed individual, and the goal is to aggregate this information.

Before describing how informational S\&C apply in such settings, we would like to contrast with approaches to crowdsourcing that models each user's contribution as a monolithic submission that has some endogenous quality, such as \citet{dipalantino2009crowdsourcing,archak2009optimal,chawla2015optimal}
There, it is impossible to integrate or aggregate user contributions and the problem is to incentivize and select one of the highest possible quality.
In these models, \emph{information} plays no role and the model is equally well-suited to incentivizing production of a high-quality \emph{good} of which only one is required; sometimes this is explicit in the motivation or model of the literature~\citep{cavallo2012efficient}.

Here, we consider cases where users have heterogeneous information and we would like to aggregate it into a final form that is more useful than any one user.
The question is how users behave strategically in revealing this information in the contest.

\paragraph{Collaborative, market-based contests for machine learning.}
\citet{abernethy2011collaborative} proposes a mechanism for machine learning contests with a prediction market structure.
This mechanism was later extended to elicit data points and to more general problems in \citet{waggoner2015market}.
While these mechanisms have appealing structure, seeming to align participants' incentives with finding optimal machine-learning hypotheses, the authors did not give results on equilibrium performance or behavior of strategic agents.
Here, we briefly describe the framework of this mechanism and how our results can apply.

In a machine learning problem, we are given a \emph{hypothesis class $\D$}.
There is some true underlying distribution $e$ of data, which is initially unknown.
In our setting we assume there is a prior belief on this distribution $e$, which distributed as a random variable $E$.
The goal is to select a hypothesis $d$ with minimum \emph{risk} $R(d,e)$ on the true data distribution.
Here, $R(d,e) = \E_{z\sim e} \ell(d,z)$ for a \emph{loss function} $\ell(d,z)$ on hypothesis $h$ and a datapoint $z$ drawn from $e$.

In the contest mechanism of \citet{abernethy2011collaborative,waggoner2015market}, the mechanism selects an initial market hypothesis $d^{(0)}$.
As in the prediction market model of Section \ref{sec:markets}, participants iteratively arrive and propose a new hypothesis $d^{(t)}$ at each time $t$.
At the end of the contest, the mechanism draws a test data point $z \sim e$ from the true distribution and rewards each participant by their improvement to the loss of the market hypothesis, \emph{i.e.} if $i$ updated the hypothesis at time $t$ from $d^{(t-1)}$ to $d^{(t)}$, then $i$ is rewarded $\ell(d^{(t-1)},z) - \ell(d^{(t)},z)$ for that update.

We observe that prediction markets are a special case of this framework: Each $e$ corresponds to some fixed observation, for instance, for a given $e$ the data point drawn is always the same $z_e$.
The risk $R(d,e)$ is therefore always equal to $\ell(d,z_e) = -S(d,e)$ for the proper scoring rule $S$ used in the market.
Thus, the above framework captures prediction markets as a special case.
However, we now show that prediction markets capture the essential strategic features of this setting.

Now, notice that in expectation over the test data $z$, this reward is equal to $R(d^{(t-1)},e) - R(d^{(t)},e)$.
Furthermore, we can define the \emph{utility} of the designer to be $u(d,e) = - R(d,e)$, that is, the negative of the risk of that hypothesis on that data distribution.
Hence, by the revelation principle, there is some proper scoring rule $S$ that is payoff-equivalent to $u$.
A prediction market with proper scoring rule $S$ is strategically identical to the above contest.
Thus, with a few small caveats, our above results apply: in a Bayesian game setting where traders have signals $A_1,\dots,A_n$ and a common prior on the distribution of signals and $E$, substitutes characterize ``all-rush'' equilibria with immediate aggregation, while complements characterize ``all-delay'' equilibria.

The caveats are (1) that the scoring rule obtained by the revelation principle will not in general be strictly proper if two beliefs about $E$ map to the same optimal hypothesis $d$; and (2) it is not guaranteed that traders can infer others' information from their trades without a condition analogous to the \emph{distinguishability} criterion of Section \ref{sec:markets}.
While these hurdles are surmountable, our purpose here is only to mention the key ideas for how the above results on prediction markets will generalize.

This connection immediately presents several questions for future work: When, in a machine-learning setting, should we expect contest participants to have substitutable or complementary information?
In particular, if agents hold \emph{data sets} and the goal is to elicit these data sets using this structure (as explored in \citet{waggoner2015market}), when should we expect data sets to be substitutable?
Furthermore, how can we \emph{design} loss functions so as to encourage substitutability and hence early participation?
This last question is discussed in Section \ref{sec:surrogate}.

\paragraph{Question-and-answer forums.}
\citet{jain2009designing} propose a model for analyzing strategic information revelation in the context of \emph{question-and-answer forums}.
Initially, some question is posed.
Participants have private pieces of information $A_1,\dots,A_n$ as in the prediction market model, and they arrive iteratively to post answers.
Unlike in the prediction market model, rather than arriving multiple times, participants may only post a single answer; however, they may be strategic about \emph{when} they post this answer.
By waiting until later, a participant may be able to aggregate information from others' answers, allowing her to post a better response.
Unlike in the prediction market setting, participants cannot ``garble'' their information.
However these are not essential differences as compared to the substitutes or complements cases of prediction markets, where in equilibrium participants do not want to garble or participate multiple times, but instead fully reveal at the time that is optimal for them (as early or as late as possible, respectively).

In the model of \citet{jain2009designing}, the asker of the question has a valuation function $\V(S)$ over subsets of the pieces of information.
\citet{jain2009designing} does not justify how such a valuation function may arise, but we can now justify this modeling decision because \emph{any} decision problem faced by the asker gives rise to some such valuation function.
In one case of \citet{jain2009designing}, the asker draws a uniform ``stopping threshold'' $t$ on $[\V(\bot), \V(A_i \join \cdots \join A_n)]$ (using our notation), and selects as the ``winning answer'' the one whose information raises her value above this threshold.
For instance, if the first two users to post are $i,j$ with signals $A_i,A_j$, and then the third user $k$ posts with signal $A_k$, and we have $\V(A_i \join A_j) < t \leq \V(A_i \join A_j \join A_k)$, then user $k$ is declared the winner.

From this model, the expected reward of a participant, which is an indicator for being declared the winner, is exactly proportional to the marginal value of her information to the information collected so far.
Thus, substitutes imply that participants' dominant strategy is to rush to participate as early as possible, while complements imply that it is dominant to wait as long as possible.
This follows from diminishing (increasing) marginal value of information.

And indeed, \citet{jain2009designing} identify substitutes and complements conditions on the information which are exactly diminishing and increasing marginal returns, with the main result as stated above. (The paper also considers several other methods of selecting the winner with more complex features, which we will avoid discussing here for simplicity.)

We would like to emphasize that, while the above discussion intentionally highlights the similarities between that work and this one, the authors do not provide any endogenous model of the information, \emph{e.g.} whether it be probabilistic and if so how it is structured, nor of the utility of the asker of the question and how this utility might arise or be related to the structure of the information.
Without such models, it does not justify why a structure might satisfy their substitutes or complements conditions (nor when/if one could expect the conditions to hold).

Our work provides answers to all of these questions.
Information may be modeled as Bayesian signals and the asker may face any decision problem.
This gives rise to a valuation function over signals that can capture the model in \citet{jain2009designing}.
Furthermore, under this model, that papers' substitutes and complements conditions used in \citet{jain2009designing} are subsumed by those proposed here.
Hence, we are able to bring this work under the same umbrella as prediction markets, the crowdsourcing contests discussed above, and the algorithmic and structural results to be discussed later.
An example of substitutes for any of these problems is an example for all of them.

\clearpage

\section{Algorithmic Applications} \label{sec:algorithmic}
Here, we investigate the implications of informational substitutes and complements on the construction and existence of efficient algorithms for information acquisition.
We first define a very general class of problems, \sigsel/, to model this problem, and show positive results corresponding to substitutes and negative results in general.
These results are obtained by showing tight connections to maximization of (submodular) set functions.
We also investigate an adaptive or online variant of the problem with similar results.

To focus on the problem of information acquisition, we abstract out the complexity of interacting with the decision problem and prior.
This differentiates these results from prior work on information acquisition, but the overall approach -- utilizing submodular set functions -- is common (even pervasive) in the literature (see~\citet{krause2012submodular}).
So we view the contribution of these results as offering a unification or explanation for successful approaches in terms of informational substitutes.

\subsection{The \sigsel/ problem} \label{sec:sigsel}
\begin{definition}
In the problem \sigsel/, one is given a decision problem $u$ and information structure on $n$ signals $A_1,\dots,A_n$; and also a family $\mathcal{F}$ of feasible subsets of signals $S \subseteq \{1,\dots,n\}$.
The goal is to select an approximately optimal $S \in \mathcal{F}$, \emph{i.e.} to
 \[ \max_{S \in \mathcal{F}} ~~ \V\left(\bigjoin_{i \in S} A_i \right) . \]
\end{definition}
For instance, suppose that an agent has a budget constraint of $B$ and each piece of information has a price tag; how to select the set that maximizes utility subject to the budget constraint?
(This is also known as a knapsack constraint.)

The \sigsel/ problem is not yet well-defined because we have not described how the input is represented.
In general, the decision problem may be hard to optimize, and the prior distribution may have support $2^{\Omega(n)}$.
Because we want to abstract out the complexity of \sigsel/ independently of the difficulty of these problems, we will assume an efficiently-queryable input.
The outline for our approach is as follows, pictured in Figure \ref{fig:sigsel-approach}.
\begin{itemize}
\item For positive results, we require an input representation that allows efficient computation of $\V\left(\bigjoin_{i\in S} A_i \right)$ for some subset $S$ of signals.
      In Section \ref{sec:computing-V}, we will discuss several such representations.
      The most natural of these we call the \emph{oracle model}.
\item For negative results, we will reduce a hard problem -- maximization of arbitrary monotone set functions given a value oracle -- to \sigsel/.
      The reduction will produce instances of \sigsel/ having a very concise and tractable input representation: The signals $A_1,\dots,A_n$ will be independent uniformly random bits, the event $E$ will be the vector $(A_1,\dots,A_n)$, and the decision problem will be immediately computable, just requiring transparent calls to the value oracle of the original maximization problem.
      This implies that any algorithm that can solve \sigsel/, under any ``reasonable'' model of input (particularly the oracle model and others we describe), can solve the instances produced by our reduction.
\end{itemize}
Note that an oracle-based approach is very general because, if we do have an instance where the input is concise and given explicitly, for example, the decision-problem optimizer is given as a small circuit, then we can just run our algorithms treating this input as an oracle, evaluating it when necessary.

We next discuss input representations for positive results, including describing the oracle model.
We will then give our positive and negative results.

\paragraph{Monotone set function maximization.}
Our results will involve relating complexity of \sigsel/ to that of maximizing some $f: 2^N \to \mathbb{R}$ where $N = \{1,\dots,n\}$ is a finite ground set.
The fact that $\V$ is increasing, \emph{i.e.} more information always helps, implies that we restrict to monotone increasing $f$: If $S \subseteq T$, then $f(S) \leq f(T)$.
Recall that submodular $f$ correspond to substitutable items, while supermodular $f$ correspond to complements.

In set function maximization, the input is often given as a \emph{value oracle} that, when given a subset $S \subseteq N$, returns in one time step $f(S)$.
When $f$ is submodular, it is known that polynomial-time constant-factor approximation algorithms exist for many types of constraints.
For instance, there are efficient $(1-1/e)$-approximation algorithms under the knapsack constraint described above~\citep{sviridenko2004note,krause2005note} and more general \emph{matroid} constraints~\citep{calinescu2011maximizing}.
On the other hand, in general set function maximization is known to be difficult information-theoretically, requiring exponentially-many oracle queries to obtain a nontrivial approximation factor even when restricting to monotone supermodular functions; we give an example in Proposition \ref{prop:max-monotone-supermodular}.

\begin{figure}[ht]
\caption{Input representations and structure of results.}
\label{fig:sigsel-approach}
\centering
\begin{subfigure}{0.9\linewidth}
\includegraphics[width=\linewidth]{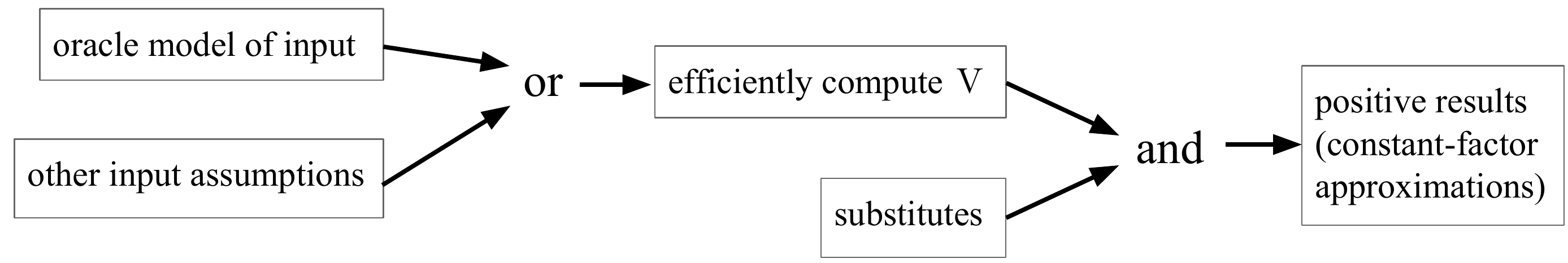}
\caption{The structure of our positive algorithmic results.
  Black arrows represent logical implication.
  Given input represented in the oracle model, or some similar models, $\V$ can be efficiently computed.
  We then show that this, along with the substitutes assumption, implies efficient algorithms for \sigsel/ for a variety of types of constraints, such as knapsack constraints.}
\end{subfigure}

\vspace{3em}
\begin{subfigure}{0.9\linewidth}
\includegraphics[width=\linewidth]{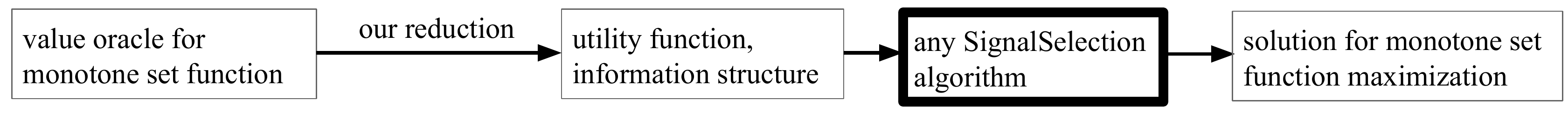}
\caption{The structure of our negative algorithmic results: a reduction to \sigsel/ from maximizing a monotone set function, given as a value oracle, under a set of constraints.
  Black arrows represent an algorithmic reduction.
  Given the value oracle, we construct a utility function and information structure.
  These are very concise and simple, allowing immediate computation of $\V$.
  Any algorithm for \sigsel/ that can accept this kind of input representation will then give a solution to the original maximization problem.}
\end{subfigure}
\end{figure}

\subsubsection{The oracle model and computing \texorpdfstring{$\V$}{V}} \label{sec:computing-V}
Here, we investigate the computation of $\V$, the value function.
We restrict attention to evaluating $\V$ at a set of signals, \emph{i.e.} the join $\bigjoin_{i\in S} A_i$ of a subset $S$ of signals.
The reason is that this case is sufficient for our positive results and is most compelling for \sigsel/.
Furthermore, no difficulty arises in how the input signal is represented, as it can always be given by a subset $S$ of $\{1,\dots,n\}$.

We begin with a case where the decision problem is specified by an oracle, but the prior $p$ is given explicitly.
Because $p$ may be exponentially large in $n$, the number of signals, we will later introduce an oracle model for $p$ as well.
\begin{proposition} \label{prop:compute-V-big}
For any decision problem and set of signals $A_1,\dots,A_n$, given an oracle for computing the associated convex $G$, we can compute $\V(\bigjoin_{i \in S} A_i)$ in time polynomial in $n$, $\prod_i |\text{Support}(A_i)|$ and $|\text{Support}(E)|$.
(This is the size of the problem in general, as the prior distribution ranges over this many outcomes $(e,a_1,\dots,a_n)$.)
\end{proposition}
\begin{proof}
We assume an oracle that computes $G(q)$ for any distribution $q$ on $E$.
Note that $G(q) = \max_d \E \left[ u(d,e) \mid e \sim q\right]$, so this is equivalent to assuming an oracle for the utility of the optimal decision for a given distribution on $E$.

We need to calculate
 \[ \V\left(\bigjoin_{i\in S} A_i\right) = \E_{\{a_i : i \in S\}} G(p_{a_i : i \in S}) . \]
Here, the expectation is over all realizations $\{a_i : i \in S\}$ of the set of signals $\{A_i : i \in S\}$, and $p_{a_i : i \in S}$ is the posterior distribution on $E$ conditioned on that set of realizations.

There are at most $\prod_{i\in S} |\text{Support}(A_i)|$ terms in the sum, and each posterior $p_{a_i : i \in S}$ can be computed as follows:
 \[ p(e | a_i : i \in S) = \frac{p(e, a_i : i \in S)}{p(a_i : i \in S)} , \]
for each $e$ in the support of $E$.
This can be computed in time polynomial in the products of the support sizes.
\end{proof}

In general, the running time of Proposition \ref{prop:compute-V-big} is exponential in $n$, the number of signals.
This is unavoidable in general as the input itself may be this large (the prior distribution ranges over exponentially many events).
Thus, it is natural to suppose that the input is given as an oracle in some fashion, making the input succinct.
We define a natural model, give the associated positive result, and discuss possible variants or weakenings.
\begin{definition}
In the \emph{oracle model} for representing a decision problem $u$ and prior $p$, one is given:
\begin{enumerate}
 \item An oracle computing the prior probability of any realization of any subset of signals.
 \item Access to independent samples from the prior distribution.
 \item An oracle computing, for a distribution $q$ on $E$, the expected optimal utility obtainable given belief $q$, namely $G(q)$.
\end{enumerate}
\end{definition}

\begin{proposition} \label{prop:compute-V-small}
In the oracle model, we can approximate $\V(\bigjoin_{i\in S} A_i)$ to arbitrary (additive) accuracy with arbitrarily high probability in time polynomial in $n$, $\sum_i \log |\text{Support}(A_i)|$, and $|\text{Support}(E)|$.
\end{proposition}
\begin{proof}
For any given $S$, we would like to approximate
 \[ \V\left(\bigjoin_{i\in S} A_i\right) = \E_{\{a_i : i \in S\}} G(p_{a_i : i \in S}) \]
up to an additive $\epsilon$ error with probability $1-\delta$.
We can abstract this problem as computing $\E Z$, with $Z$ is distributed as $G(p_{a_i : i \in S})$ where $\{a_i : i \in S\}$ is drawn from the prior.
Letting $K = \max_q G(q) - \min_q G(q)$ over all $\{q = p_{a_i : i \in S} : S \subseteq \{1,\dots,n\}\}$, we can apply a standard Hoeffding bound: The average of $m$ i.i.d. realizations of $Z$ is within $\epsilon$ of the true average with probability $1-\delta$ as long as $m$ exceeds $\frac{K^2\ln(2/\delta)}{2\epsilon^2}$.

To see that we can in fact sample $Z$: We simply draw one sample from the prior, giving us $\{a_i : i \in S\}$.
We compute the posterior conditioned on this sample as follows: For each outcome $e$ of $E$, we have
 \[ p_{a_i : i \in S}(e) = \frac{p(e, \{a_i : i \in S\})}{p(\{a_i : i \in S\})} . \]
This requires two calls to the prior computation oracle; then a call to the oracle for $G$ completes the calculation.
The running time analysis simply observes that each signal realization $a_i$ requires only $\log|\text{Support}(A_i)|$ bits to represent, and there are $n$ of them; similarly for outcomes of $E$.
\end{proof}
One would like to make weaker assumptions.
However, dropping either the assumption of independent samples or the oracle seems problematic.
Without samples, evaluating $\E G(p_{a_i : i \in S})$ seems difficult because this is a sum over exponentially many terms, and we cannot \emph{a priori} guess which terms are ``large'' or where to query the oracle for the prior.

With independent samples but no oracle for marginal probabilities, na\"{i}ve approaches break down because of the difficulty of accurately estimating conditional probabilities $p(e | \{a_i : i \in S\})$.
For instance, it may be that each outcome $\{a_i : i \in S\}$ is very unlikely, so that one cannot draw enough samples to accurately estimate the desired conditional probability using the ratio $p(e, \{a_i : i \in S\}) / p(\{a_i : i \in S\})$.
It is also of note that any distribution over the possible outcomes of the signals and of $E$, including the prior itself, in general has size $|\text{Support}(E)| \cdot \prod_{i=1}^n |\text{Support}(A_i)|$, which is exponential in $n$.
So one must avoid writing down such distributions; and any small sketch seems to quickly lose accuracy in estimating the conditional probability, which is computed from probabilities on events (and again, these probabilities may all be exponentially small even while conditional probabilities are large).

We give two further examples of how to overcome this difficulty.
The first is to assume that the prior is tractable in some way; in our case, sparse.
The second is to push the difficulties mentioned into an oracle of a different sort.
\begin{proposition} \label{prop:compute-V-small-sparse}
Suppose we are given access to an oracle for $G$ explicitly given the prior $p$; and suppose that the prior is \emph{sparse}, supported on $k$ possible outcomes $(e, a_1,\dots, a_n)$.
Then we can compute $\V(\bigjoin_{i\in S} A_i)$ in time polynomial in $n$ and $k$.
\end{proposition}
\begin{proof}
We can now assume that the prior is explicitly given as part of the input.
The expectation in the definition of $\V$ is now a sum that ranges over only at most $k$ terms. For each term, corresponding to some subset $\{a_i : i \in S\}$, we can efficiently look up $p(a_i : i \in S)$, as well as (for each $e$) $p(e,\{a_i : i \in S\})$, allowing us to compute the conditional probability $p(e \mid a_i : i \in S)$.
(Recall that our notation $p_{a_i : i \in S}$ is simply the vector of these probabilities, ranging over outcomes $e$ of $E$.)
These are all the ingredients we need to evaluate each term in the sum, which is $p(a_i : i \in S) \cdot G(p_{a_i : i \in S})$.
\end{proof}

\begin{proposition} \label{prop:compute-V-small-oracle}
Suppose we are given access to the following:
\begin{enumerate}
  \item A \emph{decisionmaking oracle} that, given a subset of signal realizations, returns an optimal decision $d^*$; and
  \item an oracle for evaluating the utility $u(d,e)$ of decision $d$ when nature's event $E=e$; and
  \item access to independent samples of $(e,a_1,\dots,a_n)$ from the prior distribution.
\end{enumerate}
Suppose that $u(d,e)$ is bounded.
Then we can approximate $\V(\bigjoin_{i\in S} A_i)$ to arbitrary (additive) accuracy with arbitrarily high probability in time polynomial in $n$ and $\sum_i \log |\text{Support}(A_i)|$ and $\log |\text{Support}(E)|$.
\end{proposition}
\begin{proof}
Again, given $S$, we wish to approximate $f(S)$, which is equal to
 \[ \V(\bigjoin_{i\in S} A_i) = \E_{e,a_1,\dots,a_n} u(d^*(a_i : i \in S), e) , \]
where $d^*(a_i : i \in S)$ is the optimal decision conditioned on observing signals $\{a_i : i \in S\}$.
Again, assuming that $u(d,e)$ lies in a bounded range of size $K$, the same Hoeffding bound applies: By drawing $m$ i.i.d. samples from the prior and then calling the oracles to obtain $d^*$ and $u$, we can approximate this expectation to arbitrary additive error with arbitrarily high probability, for a suitable (polynomial-sized) $m$.
\end{proof}
This range of results gives some evidence that in general, if the decision problem has some sort of succinct representation outside of the oracle model, then we may still hope to compute $\V$.
In fact, the problem we will construct for our negative result, Theorem \ref{thm:f-to-sigsel}, will fit the model of Proposition \ref{prop:compute-V-small-oracle}, where it is trivial to find and evaluate the optimal decision.

\subsubsection{Positive and negative results via reductions}
\paragraph{Positive results.}
As a corollary of the above reductions, we are able to reduce \sigsel/ to set function maximization, netting positive results especially in the case of substitutes.
\begin{theorem} \label{thm:alg-subs-pos}
\thmalgsubspos
\end{theorem}
\begin{proof}
Given an instance of \sigsel/ with signals $A_1,\dots,A_n$, we construct a monotone increasing set function $f: 2^{\{1,\dots,n\}} \to \mathbb{R}$ via $f(S) = \V(\bigjoin_{i\in S} A_i)$.
If signals are weak substitutes, then $\V$ is submodular on the signal lattice, and hence $f$ is submodular (as well as monotone).
We can therefore apply known algorithms for submodular maximization, in particular, \citet{nemhauser1978analysis, sviridenko2004note, calinescu2011maximizing}.
Now, there is a subtlety: Under some of the models we propose, particularly the oracle model, we do not compute $f$ exactly but instead can only guarantee arbitrarily high accuracy with arbitrarily high probability.
However, it is well known (\emph{e.g.} \citet{kempe2003maximizing}) and proven (\emph{e.g.} \citet{krause2005note}), that these algorithms still give guarantees when we have a high-probability, high-accuracy guarantee on evaluations of $f$.
The key point is that we can still evaluate, with arbitrary accuracy, the gradient or marginal contributions of each element\footnote{Recent work (in preparation) considers the question of how much accuracy in evaluating $f$ is required, showing negative results when (roughly) accuracy is worse than $\frac{1}{\sqrt{n}}$; but under the oracle model we can evaluate $f$ with error an arbitrary inverse polynomial with only an exponentially-small probability of error (simply via a Hoeffding bound), in which regime it is well-known that this problem does not arise.}.
\end{proof}

\paragraph{Discussion of approximate oracles.}
One would like a robustness guarantee of the following sort: Even if we do not have an oracle that exactly optimizes the decision problem, suppose we do have an oracle returning, say, a decision whose expected utility is within a constant factor of optimal.
Then give a good algorithm for \sigsel/ in the substitutes case, with an appropriately-decreased guarantee.

Unfortunately, it seems that this kind of result is unlikely without deeper investigation and further work.
To see the challenge, imagine that an adversary is allowed to design the oracle subject to a constraint of some approximation ratio.
Then our problem essentially reduces to submodular maximization with noisy or approximate value oracles, where the noise may be adversarially chosen subject to this approximation constraint.
Unfortunately, recent work on these kinds of problems have shown them to be difficult in general~\citep{singer2015information,balkanski2015limitations,hassidim2016submodular}.
This seems like a very difficult barrier, but perhaps future work can leverage the structure of decision problems in some way to make progress in this direction.

\paragraph{Negative results.}
We show that in general, \sigsel/ is as difficult as optimizing general monotone set functions subject to constraints.
In fact, this holds even for an easy special case of \sigsel/ where all signals are independent uniformly-random bits and the decision problem is trivial to optimize (the solution is essentially to list all of the outcomes of signals you have observed).

To do so, we give a reduction in the opposite direction: Given $f$, we construct a decision problem and prior distribution such that $\V(\bigjoin_{i\in S} A_i) = f(S)$.
Hence, any algorithm for \sigsel/ gives an algorithm for optimizing $f$.
In terms of input representation, while $f$ may require exponential space to represent explicitly, our reduction is essentially as useful as one could hope for in this respect, creating a trivial wrapper around a value oracle for $f$.

It seems that any reasonable algorithm one might propose for optimally selecting sets of signals should be able to handle such a tractable input.
Hence, this is a strong negative result that, in general, optimization over signal sets is just as hard as over item sets.
\begin{theorem} \label{thm:f-to-sigsel}
\thmftosigsel
\end{theorem}
\begin{proof}
Each signal $A_i$ will be a uniform independent bit, and the event $E$ of nature will consist of the vector of realized $a_i$s (a binary string of length $n$).
Intuitively, the idea is that having observed $A_i$, regardless of whether its realization is $0$ or $1$, corresponds to having item $i$ in the set $S$, while not having observed $A_i$ (hence having a uniform belief over $A_i$) corresponds to $i \not\in S$.

The decision problem will look like a scoring rule, but it will be for predicting the \emph{mean} of $E$ rather than predicting a probability distribution over $E$.
This is good news in terms of the representation size: predicting the mean of $E$ requires only reporting a vector in $[0,1]^n$, while a general probability distribution over all outcomes of $E$ has support size up to $2^n$.

Formally, it turns out that the scoring rule characterization applies equally well to constructing proper rules for predicting the mean of a random variable such as $E$.
Specifically, given any convex function $F: \mathbb{R}^n \to \mathbb{R}$, one can construct a scoring rule $R: \mathbb{R}^n \times \mathbb{R}^n \to \mathbb{R}$.
Using the notation $R(r;q) = \E_{e\sim q} R(r,e)$, we have the key scoring rule property that $R(r;q)$ is maximized at $r = \E_{e\sim q} e$, where it equals $F(\E_{e\sim q} e)$.
A quick proof: Given $F$, define $R(r,e) = F(r) + \langle F'(r), e-r \rangle$ where $F'(r)$ is a subgradient at $r$.
We have $R(\E_{e\sim q} e; q) = F(\E_{e\sim q} e)$ as desired.
Now, note that $D_F(\E_{e\sim q} e;r) = F(\E_{e\sim q}) - R(r;q)$, where $D_F$ is the Bregman divergence of $F$; since Bregman divergences are nonnegative, this proves that $R(r;q)$ is maximized at $R(r;q) = R(\E_{e\sim q}; q) = F(\E_{e\sim q})$. 

Hence, the roadmap is as follows.
\begin{enumerate}
 \item Construct a function $F: [0,1]^n \to \mathbb{R}$.
 \item Verify that $F$ is convex.
       This implies that there is a decision problem (namely, predicting the mean of $E$) where the expected utility for predicting $\mu$ when one's true expectation is $\mu$ equals $F(\mu)$.
 \item Verify that, for any subset $S$ of $\{1,\dots,n\}$ and for any set of realizations $\{a_i : i\in S\}$, we have $F(\E[e \mid a_i : i \in S]) = f(S)$.
 \item Note this implies that, for any $S$, we have $\V(\bigjoin_{i \in S} A_i) = f(S)$.
 \item Check that, given a value oracle for $f$, we can efficiently compute any quantities of interest (the posterior distribution on $E$, the posterior expectation of $E$, $F(\E[ e \mid a_i : i \in S])$, $\V(\bigjoin_{i \in S} A_i)$).
\end{enumerate}

\paragraph{(1)}
The construction of $F$ is recursive on the dimension $n$.
It is ugly, being discontinuous at its boundary.
However, drawing some pictures should convince the reader that $F$ can be ``smoothed'' to a more reasonable, continuous convex function.
For a base case of $n=1$, on $[0,1]$ we let $F(r) = f(\emptyset)$ on the interior where $r \in (0,1)$, and $F(0) = F(1) = f(\{1\})$.
(That is, $f$ evaluated at the set consisting of the item.)
Note that, because $f$ is monotone increasing, \emph{i.e.} $f(\{1\}) \geq f(\emptyset)$, $F$ is convex.
The discontinuity implies that, for the associated proper scoring rule for the mean $R$, we must have $R(\mu,e) = -\infty$ whenever $\mu$ lies at an endpoint and $e$ is in the interior.
But again, any smoothing of $F$ to be continuous will remove this property.

On $[0,1]^n$, we let $F(r) = f(\emptyset)$ for $r$ in the interior of the hypercube.
That is, $F$ is constant on its interior.
Furthermore, we have $F(\E_{e\sim p} e) = f(\emptyset)$, where $p$ is the prior, hence $\V(\bot) = f(\emptyset)$.

Now we define $F$ on its boundary.
Consider any face of the $n$-dimensional hypercube.
Each face corresponds to a particular setting of some $A_i$, either to $0$ or $1$, by the coordinate whose value is constant on that face.
For instance, $A_i=1$ corresponds to the face consisting of the set of $r \in [0,1]^n$ where $r$'s $i$th coordinate equals one.
On both of the faces corresponding to $A_i$, $F(r)$ is defined to be $F(r) = F_{\{i\}}(r_{-i})$, where $r_{-i}$ is $r$ with the $i$th coordinate removed, and the function $F_{\{i\}}$ is defined recursively as follows.
Consider the set function $f_{\{i\}}: 2^{\{1,\dots,n\}\setminus\{i\}} \to \mathbb{R}$ with $f_{\{i\}}(S) = f(S \cup \{i\})$.
Then let $F_{\{i\}}$ be the result of our construction applied to $f_{\{i\}}$.
This completes the definition of $F$.

Verbally, for each face of the hypercube, we have fixed some $i$ to be in the set $S$ passed to $f$, and considered the resulting submodular function $f_{\{i\}}$ on the remainder of $\{1,\dots,n\}$.
The value of $F$ on the interior of that face will be $f(\{i\})$, by the recursive construction.
To picture $F$ and convince ourselves that it is well-defined, consider the intersection of the faces corresponding to, say, $A_i = 1$ and $A_j = 0$.
This is a lower-dimensional face consisting of all points on the hypercube whose $i$th bit equals $1$ and $j$th bit equals $0$.
On the interior of this face (\emph{i.e.} no other bits are equal to $0$ or $1$), $F$ has value $f(\{i,j\})$.
And so on all the way ``out'' to the corners $r$ of the hypercube, where $r \in \{0,1\}^n$; at all of these, $F(r) = f(\{1,\dots,n\})$.

\paragraph{(2)}
We prove $F$ is convex on the hypercube by induction on $n$, with the base case $n=1$ already observed above.
For the inductive step, note again the key point: by monotonicity of $f$, if $r$ is on the boundary of the hypercube and $s$ is in the interior then $F(r) \geq F(s)$.
Consider any two points $r,s \in [0,1]^n$, and break into cases.
If both points lie in the interior, then because $F$ is constant there, $F(\alpha r + (1-\alpha)s) = \alpha F(r) + (1-\alpha)F(s)$ for any $0 \leq \alpha \leq 1$.
If one point lies in the interior and one on the boundary, then any convex combination of the two lies in the interior.
$F$ is constant in the interior and weakly larger on the boundary, so the convexity inequality is satisfied.
If the points lie in different faces, then again any convex combination lies in the interior.
Finally, if the points lie in the same face, then $F$ coincides with some $F_{\{i\}}$ on $r$ and $s$, and $F_{\{i\}}$ is convex by inductive hypothesis.

\paragraph{(3)}
We now verify that $F$, when applied to the expected value of $E$ given the realizations of signals corresponding to $S$, is equal to $f(S)$.
Consider any set of realizations $\{a_i : i \in S\}$ and let $\mu = \E[e \mid a_i : i \in S]$.
If $S = \emptyset$, then by construction $\mu = (0.5,\dots,0.5)$ and $F(\mu) = f(\emptyset)$.
Otherwise, $\mu$ is the vector where each entry $i$ is equal to $a_i$ if $i \in S$ and $0.5$ otherwise; this follows from the i.i.d. distribution of the $A_i$.
Hence, $\mu$ lies in the interior of the (low-dimensional) face of the hypercube corresponding to the realizations $\{a_i : i \in S\}$, hence $F(\mu) = F_{S}(\mu_{-S}) = f_S(\emptyset) = f(S)$.
Here the notation $f_S$ is the function obtained from $f$ by fixing $S$; $F_S$ is the corresponding recursively constructed function on $[0,1]^{n-|S|}$; and $\mu_{-S}$ is $\mu$ obtained by removing all coordinates.

\paragraph{(4)}
Use the notation $\mu_{a_i : i \in S} = \E_e \left[ e \mid a_i : i \in S\right]$.
Since step (3) holds for all realizations of a given set of signals, in particular (where $R$ is the scoring rule corresponding to $F$):
\begin{align*}
 \V\left(\bigjoin_{i \in S} A_i\right)
  &= \E_{a_i : i \in S} \max_{d \in [0,1]^n} \E_e \left[ R(d,e) \mid a_i : i \in S \right]  \\
  &= \E_{a_i : i \in S} \max_{d \in [0,1]^n} R(d ; p_{a_i : i \in S})  & \text{definition of notation $R(d;q)$}  \\
  &= \E_{a_i : i \in S} R(\mu_{a_i : i \in S} ; p_{a_i : i \in S})     & \text{properness of $R$}  \\
  &= \E_{a_i : i \in S} F(\mu_{a_i : i \in S})                         & \text{construction of $F$ and $R$}  \\
  &= \E_{a_i : i \in S} f(S)                                           & \text{step (3) of proof}  \\
  &= f(S) .
\end{align*}

\paragraph{(5)}
Given any set of signals $S$, one can immediately compute $\V(S) = f(S)$ by a call to the value oracle for $f$.
Given their realizations $\{a_i : i \in S\}$, the posterior distribution is simply uniform on those $A_j$ with $j \not\in S$ (and of course has each $A_i = a_i$ with probability one for $i \in S$).
This induces the posterior distribution over $E$ (which can thus be concisely represented, even though there are $2^n$ possible outcomes in general), as well as the posterior expectation of $E$.
Evaluating $F(r)$ at an arbitrary point $r$ can be done quickly: Let $S$ be the subset of coordinates on which $r$ is equal to either $0$ or $1$; then $F(r) = f(S)$, as $r$ lies on the interior of a face corresponding to $S$.
This requires just a single call to the oracle for $f$.
Evaluating $R(d,e)$ can thus be done in time polynomial in $n$ as well; the only additional step required is picking a subgradient of $F$ at each point $d \in [0,1]^n$, and there is only one choice at each point (as noted above, unless $F$ is smoothed, this construction does require $R(\mu,e) = -\infty$ if $\mu$ lies on a face and $e$ in the interior).
Finally, the optimal decision for a given expectation of $e$ is just that expectation.
\end{proof}

\paragraph{Explicit negative results.}
First, our reduction has implications even for the substitutes case.
Monotone submodular maximization to a better factor than $1-1/e$, even under a simple cardinality constraint, requires an exponential number of value-oracle calls~\citep{nemhauser1978best} and, even given a concise explicit input, is NP-hard~\citep{feige1998threshold}.
Because our reduction preserves marginal values, a submodular set function reduces via Theorem \ref{thm:f-to-sigsel} to weak substitutes.
This implies that our $(1-1/e)$ approximation cannot be improved without a stronger assumption than substitutability (or \emph{e.g.} a polynomial-time algorithm for SAT).
\begin{corollary} \label{cor:subs-better-approx-hard}
Even when signals are weak substitutes with a cardinality constraint, achieving a strictly better approximation than $1-1/e$ to \sigsel/ requires an exponential number of oracle queries.
\end{corollary}

Without the substitutes/submodularity condition, it is well known that maximizing general monotone set functions is difficult given only a value oracle, although such negative results are not always explicit in the literature.
For instance and for concreteness, one can even restrict to monotone supermodular functions and the simple problem of maximizing $f$ subject to a cardinality constraint (select any set $S$ of at most $k$ elements, for some given $0 \leq k \leq n$).
In this case, there does not seem to be a negative result in the literature despite this hardness being well-known (see \citep{usul2016maximizing}), so we give an explicit one here to illustrate the challenge.
\begin{proposition} \label{prop:max-monotone-supermodular}
For any algorithm for monotone supermodular maximization subject to a cardinality constraint, if the algorithm makes a subexponential number of value queries, then it cannot have a nonzero approximation ratio.
\end{proposition}
\begin{proof}
We construct a simple family of supermodular functions.
Let $k$ be the cardinality constraint, \emph{i.e.} maximum cardinality of any feasible set.
Let $f(S) = 0$ if $|S| \leq k$ and otherwise let $f(S) = |S| - k$.
Pick one special set $S^*$ of size $k$, uniformly at random from such sets, and let $f(S^*) = 0.5$.
Hence the optimal solution to the problem is to pick $S^*$ with solution value $0.5$.

This function is supermodular: Given any element $i$, the marginal contribution of $i$ to $S$ is $0$ if $|S| < k-1$, then either $0.5$ or $0$ for $|S|=k-1$, then either $0.5$ or $1$ for $|S|=k$, then $1$ for $|S| > k$.
Because this marginal contribution is increasing, $f$ is supermodular.
Now any algorithm making subexponentially (in $k$) many queries cannot, except with vanishing probability, guess which set of size $k$ is $S^*$, so it will with probability tending to $1$ select some other set of size $\leq k$, which has solution value $0$.
\end{proof}

\begin{corollary} \label{cor:comps-alg-hard}
\corcomplementsalghard
\end{corollary}

\subsubsection{Adaptive \sigsel/}
We now define an adaptive version of the problem and show that substitutability implies positive results here as well.
This problem has already been studied in very similar settings with a very similar approach by \citet{asadpour2008stochastic,golovin2011adaptive}.
There is a significant difference in that these works aimed to maximize a set function over items in tandem with observations about those items.
Our model is a significant generalization in that it considers arbitrary kinds of observations and an arbitrary decision problem.
However, the goal of this section is not to claim originality or generality of solution, but only to demonstrate that this problem can also be viewed through the lens of substitutability.
\begin{definition}
In Adaptive \sigsel/, one is given a decision problem $u$ and information structure on $n$ signals $A_1,\dots,A_n$; and also a family $\mathcal{F}$ of feasible subsets of signals $S \subseteq \{1,\dots,n\}$.
An algorithm for this problem is a policy that first selects a signal $i_1$ to inspect; then observes the realization of $A_{i_1}$; then depending on that realization, selects a signal $i_2$ to inspect, observing the realization of $A_{i_2}$, and so on.
The algorithm must guarantee that the total set $S = \{i_1,i_2,\dots\}$ of signals selected is in $\mathcal{F}$.
After ceasing all inspections, the algorithm outputs some decision $d$.
The goal is to maximize $\E\left[ u(d,e) \right]$ over the distribution of signals and $E$ as well as any randomness in the algorithm.
\end{definition}
We note that the definition of \sigsel/ abstracted out the process of deciding an action $d$ given the realizations of the signals, as we assumed that it was known how to optimize the utility function.
Here, it seems cleaner to integrate the choice of the decision with the algorithm; but this is not the fundamental difference between the two settings.
The fundamental difference is adaptivity: In Adaptive \sigsel/, the algorithm may observe the outcome of the first selected signal before deciding which signal to select next, and so on.

Here we need a somewhat stronger substitutes condition.
\begin{definition} \label{def:ptwise-subs}
A set of signals on a subsets lattice $\Lat$ are \emph{pointwise weak substitutes} if, for each $A' \preceq A$ and $B$ in $\Lat$, and for each of the outcomes $a' \in A',a \in A$ in the support of the prior, letting $Q',Q$ be the posterior distributions conditioned on $A'=a'$ and $A=a$ respectively,
 \[ \V^{u,Q'}(B) - \V^{u,Q'}(\bot) \geq \V^{u,Q}(B) - \V^{u,Q}(\bot) . \]
In other words, the left side is the marginal value of $B$ given the observation $a'$, and the right is its marginal value given observation $a$.
\end{definition}
Verbally, signals are pointwise substitutes if having observed more information never increases the expected marginal value of a signal.
The key difference is that for ``regular'' substitutes, the condition can be written as follows:
 \[ \E_{a'} \left[ \V^{u,Q'}(B) - \V^{u,Q'}(\bot) \right] \geq \E_a \left[ \V^{u,Q}(B) - \V^{u,Q}(\bot) \right]. \]
In other words, signals are substitutes if \emph{expecting} to observe more information never increases the expected marginal value; they are pointwise substitutes if observing more information never increases expected marginal value.

Our goal here is to illustrate that techniques from submodular maximization extend naturally.
We focus on the simple cardinality constraint case and show that the greedy algorithm for monotone submodular maximization~\citep{nemhauser1978analysis} also gives good guarantees in this adaptive setting, when we have a strong substitutes condition.
\begin{algorithm}
\caption{Greedy algorithm for Adaptive \sigsel/ under cardinality constraints.}
\label{alg:adaptive-greedy}
\begin{algorithmic}[1]
  \State{\textbf{Input:} Decision problem $u,E,A_1,\dots,A_n,p$, in oracle model; cardinality constraint $k$}
  \State{Let $S_0 = \emptyset$}
  \State{Let $q_0 = p$, the prior}
  \For{$t \in \{1,\dots,k\}$}
    \State{Let $i_t = \arg\max_{j\not\in S} \V^{u,q_{t-1}}(A_j)$}
	\State{Let $S_t = S_{t-1} \cup \{i_t\}$.}
    \State{Select $A_{i_t}$ and observe $a_{i_t}$}
    \State{Let $q_t$ equal the Bayesian posterior $p_{a_j : j \in S_t}$}
  \EndFor
  \State{\textbf{Output:} $\arg\max_d \E_{e\sim q_t} u(d,e)$}
\end{algorithmic}
\end{algorithm}
\begin{theorem} \label{thm:adaptive-subs-pos}
When signals are pointwise weak substitutes, Algorithm \ref{alg:adaptive-greedy} (Adaptive Greedy) obtains a $(1-1/e)$-approximation algorithm for Adaptive \sigsel/ under cardinality constraints in the oracle model.
\end{theorem}
\begin{proof}
The usual proof essentially goes through.
We utilize notation from Algorithm \ref{alg:adaptive-greedy} (adaptive greedy).
The expected performance of $S_t$, the set constructed at time $t$, is $\V^{u,P}(S_t)$ (we use this notation as shorthand for evaluation at the join of all signals in $S_t$).

Let $V^*$ be the expected performance of the optimal policy.
We will upper-bound $V^*$ for each $t$ by the performance of a hypothetical algorithm that first selects $S_t$, then selects all $k$ of the signals selected by the optimal policy.
For a fixed $S_t$, such a hypothetical algorithm can obtain, in expectation, at most the performance of $S_t$ plus $k$ times the maximum expected marginal contribution of any signal to $S_t$.
This follows from pointwise substitutes: as the algorithm acquires more information, the expected marginal contribution of a signal does not increase.
So, taking the expectation over this condition, which holds for each realization of $S_t$:
 \[ V^* \leq \V^{u,P}(S_t) + k \Big(\V^{u,P}(S_{t+1}) - \V^{u,P}(S_t)\Big) . \] 
This is exactly the inequality that produces the approximation ratio in submodular set function case.
The inequality implies by induction that $\V^{u,P}(S_k) \geq \left(1 - \left(1-\frac{1}{k}\right)^k \right)V^*$, and this approximation ratio is always better than $1-1/e$.
\end{proof}



\clearpage

\section{Structure and Design} \label{sec:structural}
In this section, we give some initial investigations into the structure of substitutes and complements.
We focus on two kinds of questions: Understanding what types of information structures may be generally substitutes or complements for a broad class of decision problems; and understanding how to design a proper scoring rule or other decision problem so as to impose substitutability on a given signal structure.

In terms of contributions of this paper to these problems and starting points for future work, it is also worth recalling the relatively diverse set of equivalent definitions of informational S\&C derived in Section \ref{sec:characterizations}.
There, we showed that substitutes can be defined in terms of submodularity, generalized entropies, or divergences.

\subsection{Universal substitutes and complements} \label{sec:universal}
An ideal starting point would be characterizing information structures that are always substitutable or complementary; we term these ``universal''.
We note that \citet{borgers2013signals} investigated a two-signal variant of this problem, with a definition essentially corresponding to weak substitutes on those two signals, with some results of the same flavor.

Intuitively, there are ``trivial'' cases that satisfy this universality criterion, at least for weak S\&C.
An example of trivial weak substitutes is the case where $A_1 = A_2 = \cdots = A_n$.
Here, after observing any one signal, all of the others do not change the posterior belief at all.
An example of trivial weak complements is the case where each $A_i$ is an independent uniformly random bit and $E = A_1 \oplus A_2 \oplus \cdots \oplus A_n$, the XOR of all the bits.
Here, any subset of $n-1$ signals does not change the posterior belief at all compared to the prior, but the final signal completely determines $E$.
\begin{definition} \label{def:universal}
Given an information structure $E,A_1,\dots,A_n$ with prior $P$, we term $A_1,\dots,A_n$ \emph{universal weak substitutes} if they are weak substitutes for every decision problem.
Universal moderate/strong substitutes and weak/moderate/strong complements are defined analogously.
We term the signals \emph{trivial substitutes} if for every realization of $a_1,\dots,a_n$ in the prior's support, $p_{a_i} = p_{a_1,\dots,a_n}$ for all $i$ and \emph{trivial complements} if for all realizations $a_1,\dots,a_n$ in the support, $p_{\{a_j : j \neq i\}} = p$ for all $i$.
We term them \emph{somewhat trivial} if the prior is a mixture distribution that is equal to a trivial structure with some probability, and some other arbitrary other structure with the remaining probability.
\end{definition}

\begin{figure}[ht]
\caption{\textbf{Universal substitutes and complements.}
  In each plot, there are two binary signals $A_1,A_2$ and a binary event $E$.
  The $x$-axis is the probability of $E$.
  $G$ is the expected utility function for some decision problem (different in each plot).
  The circles correspond to the value of no signals (black), one signal (blue), and two signals (red).
  The blue braces measure the distance from the prior to posteriors on one signal; red measure additional distance to the posterior on two.}
\label{fig:universal}

\begin{subfigure}{1.0\textwidth}
\centering
\includegraphics[width=0.6\textwidth]{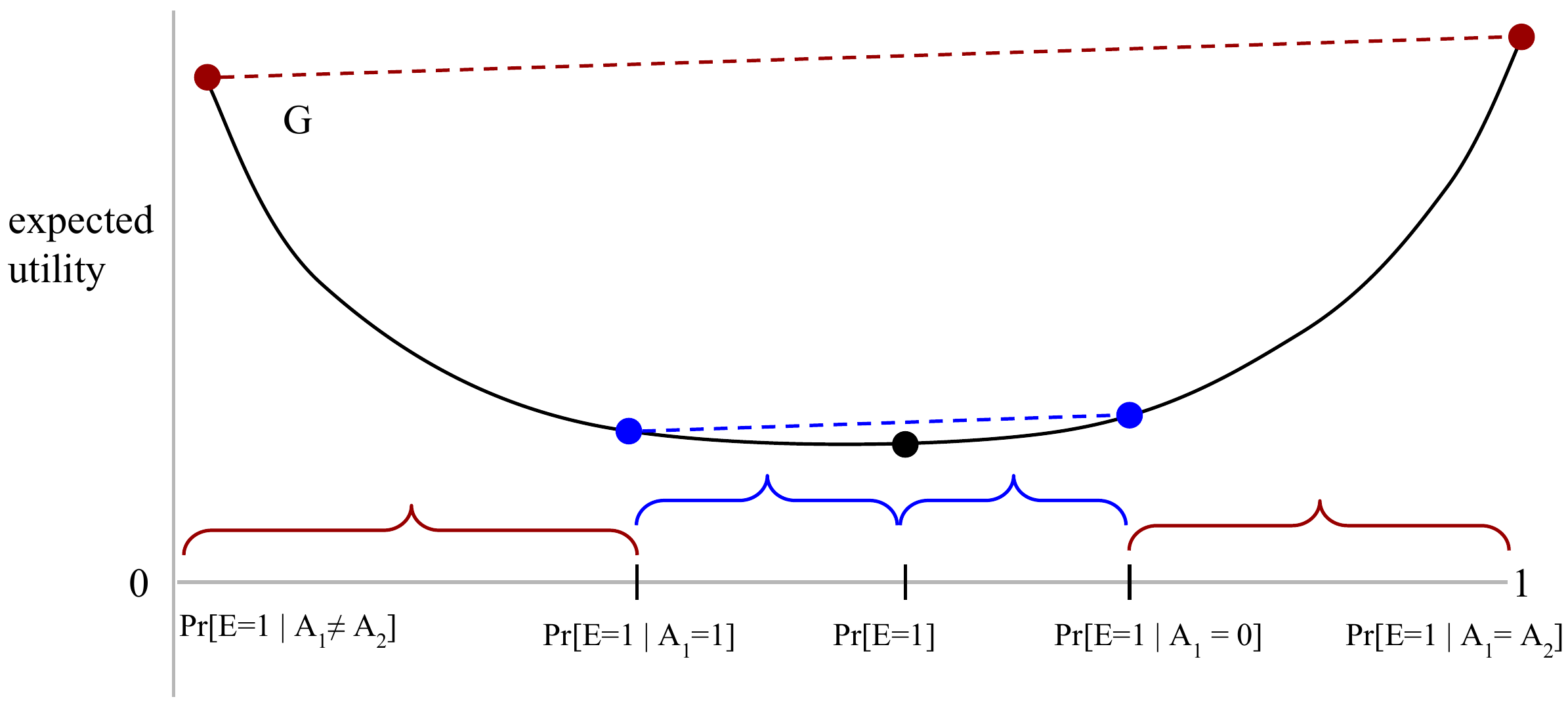}
\caption{\textbf{Curvature increases complementarity.}
  Here $A_1,A_2$ are i.i.d. noisy bits and $E = A_1 \oplus A_2$, the XOR.
  Because $G$ is convex, it is more extreme at more extreme beliefs, which tend to correspond to having multiple signals.
  Here the marginal value of a second signal is much higher than that of a first; complementarity.}
\end{subfigure}

\vspace{24pt}

\begin{subfigure}{1.0\textwidth}
\centering
\includegraphics[width=0.6\textwidth]{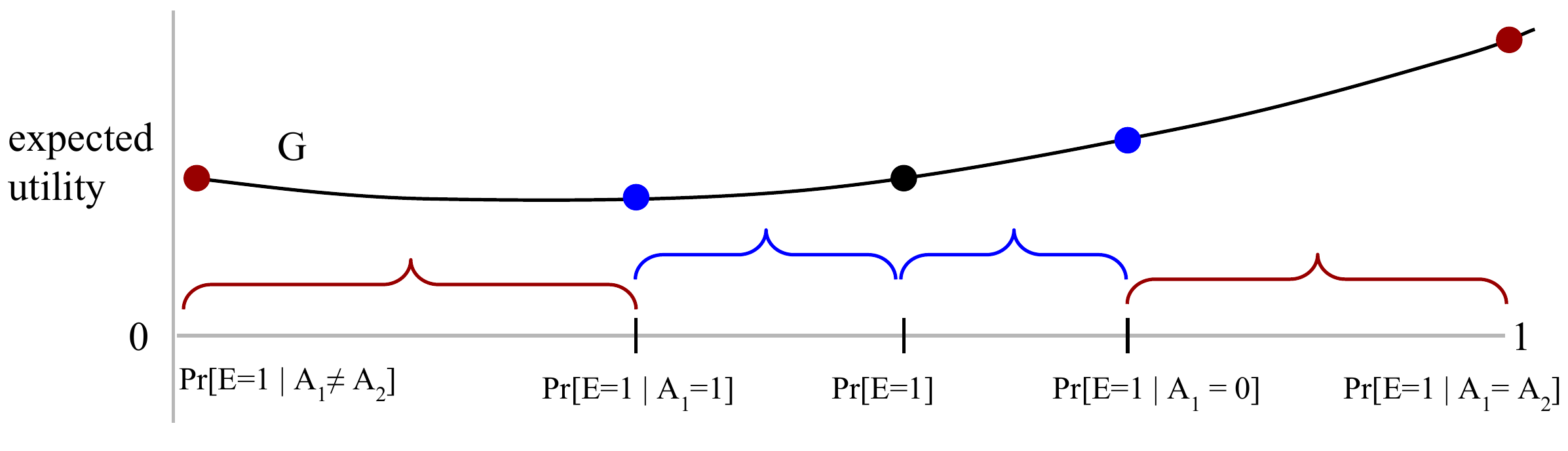}
\caption{\textbf{Flattening $G$ cannot remove complementarity.}
  For these signals (the same as in (a)), the Bayesian update on the second signal is always larger than on the first.
  Intuitively, even for flat $G$, this structure will always exhibit complementarity.
  These signals are universal complements.}
\end{subfigure}

\vspace{24pt}

\begin{subfigure}{1.0\textwidth}
\centering
\includegraphics[width=0.6\textwidth]{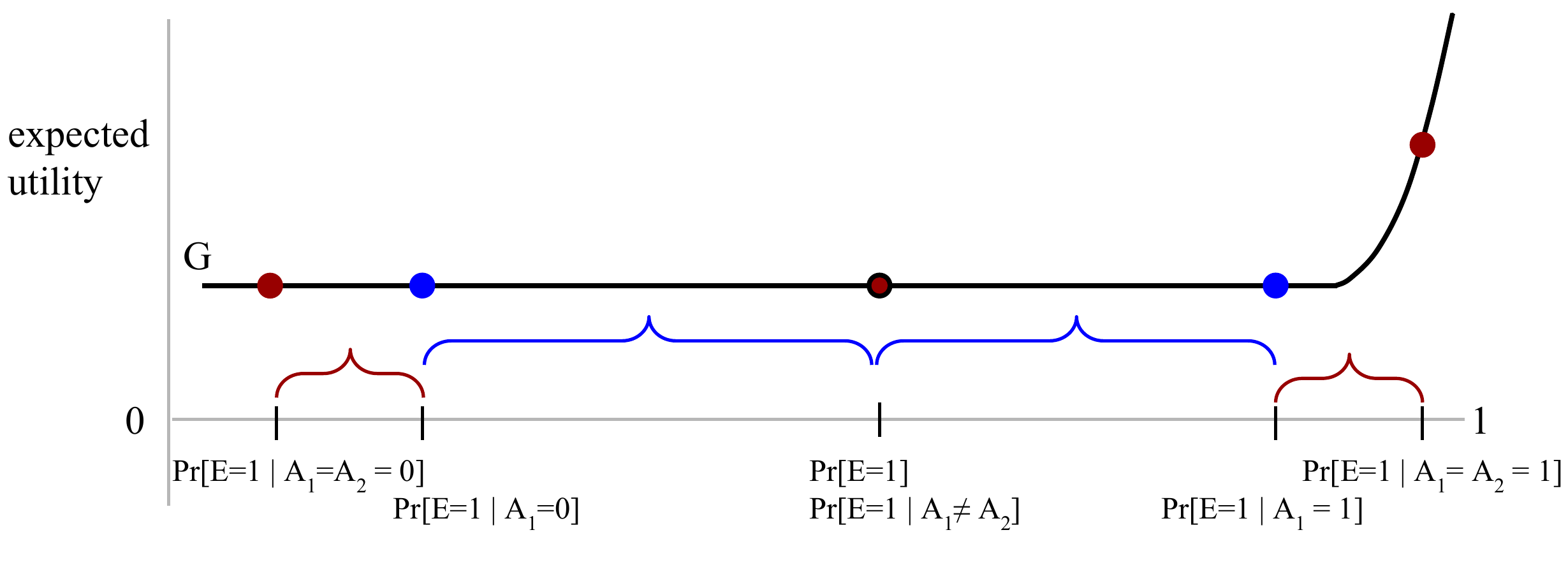}
\caption{\textbf{Curvature destroys substitutability.}
  Here $E$ is a uniform bit conditioned on which each $A_i = E$ independently with large probability.
  These are often substitutes, but by introducing high curvature at an extreme point, we create a problem where useful decisions can only be made with access to multiple signals.
  The first signal has no marginal value over the prior, while the second signal has some marginal value, so they are now complements.
  In general, this argument shows that universal substitutes must involve a ``triviality'' ruling out this construction.}
\label{subfig:univ-subs}
\end{subfigure}

\end{figure}

\begin{proposition} \label{prop:subs-mix-trivial}
If $A_1,\dots,A_n$ are universal weak substitutes, then they are somewhat trivial.
Furthermore, their ``trivial'' component is more informative than the nontrivial component, in the following sense.
Let $X_i \subseteq \Delta_E$ be the convex hull of $\{p_{a_i} : a_i \in A_i\}$, and let $Y \subseteq \Delta_E$ be the convex hull of $\{p_{a_1,\dots,a_n} : a_1 \in A_1,\dots,a_n \in A_n\}$.
If $A_1,\dots,A_n$ are universal substitutes, then $X_i = Y$ for all $i$.
\end{proposition}
\begin{proof}
Consider any information structure that does have $X_i = Y$ for all $i$.
We give a decision problem for which it sometimes satisfies strictly increasing marginal value.
To do so, by Corollary \ref{cor:decision-captured-by-G}, it is enough to construct a convex function $G$ with appropriate structure, for which we then know a decision problem (namely, a scoring rule, Fact \ref{fact:scoring-char}) exists.
The idea is pictured in Figure \ref{subfig:univ-subs}.

Let $X$ be the convex hull of $\{p_{a_i} : i=1,\dots,n; a_i \in A_i\}$, that is, of all possible posterior beliefs conditioned on one signal.
Let $G(q)$ be zero for $q$ in the convex hull of that set of beliefs (note that the prior $p$ must be in the convex hull) and $G(q)$ be increasing outside of this convex hull.
Then for any one signal $A_i$, we must have $\V(A_i) = \V(\bot) = 0$ as $\V(A_i) = \E_{a_i} G(p_{a_i})$.
But by assumption, there exist posterior beliefs for multiple signals that fall outside this convex hull.
This follows because each $p_{a_i} = \E_{a_j} p_{a_ia_j}$ for any $j \neq i$, so unless $p_{a_ia_j} = p_{a_i}$ for all $j$, there are cases where some posterior belief falls outside the convex hull mentioned.
(And if $p_{a_ia_j} = p_{a_i}$ always, then repeat the argument for triples of signals $a_i,a_j,a_k$, and so on; by nontriviality, the argument will succeed at some point.)

Now for these posterior beliefs falling outside the convex hull, they occur with positive probability and have a positive value of $G$, hence the marginal value of additional signals because strictly positive at some point, while the marginal value of the first signal was $0$.

To see that this implies a mixture containing trivial substitutes, note that the convex hulls $X_i = X_j$ for all $i,j$, so in particular the corners of these convex hulls must be points where $p_{a_i}$ is equal, in the case where $A_i = a_i$, to every posterior belief conditioned on any number of signals.
\end{proof}

\begin{theorem} \label{thm:no-universal-moderate-subs}
All universal moderate substitutes are trivial.
(Hence, the same holds for universal strong substitutes.)
\end{theorem}
\begin{proof}
Any universal moderate substitutes must be universal weak substitutes as well; so we know via the proof of Proposition \ref{prop:subs-mix-trivial} that any candidates must have a mixture with trivial substitutes.
However, we can let $A'$ be a signal determining whether that component of the mixture has occurred.
Then one implication of moderate substitutes is that, conditioned on $A'$, all signals are weak substitutes.
Therefore, universal substitutes implies that even conditioned on $A'$, the prior is a mixture containing a trivial component.
Repeating the argument gives that the entire prior must consist only of trivial components.
\end{proof}

\begin{proposition} \label{prop:universal-comps}
Consider a binary event $E \in \{0,1\}$ prior probability $p = \Pr[E=1]$, and signals $A_1,A_2$ with posterior probabilities $p_{a_1} = \Pr[E=1|A_1=a_1]$, and so forth.
Suppose that:
\begin{enumerate}
 \item $\max_{a_1} \|p-p_{a_1}\| \leq r$ and $\max_{a_2} \|p-p_{a_2}\| \leq r$; and
 \item $\min_{a_1,a_2} \|p - p_{a_1,a_2}\| \geq 2r$.
\end{enumerate}
Then $\{A_1,A_2\}$ are universal moderate complements.
\end{proposition}
\begin{proof}
The idea is that posterior beliefs on multiple signals usually tend to lie ``farther'' from the prior than those on single signals, because at least in some cases, more information leads to more certain (hence extreme) beliefs.
Convex functions already tend to give larger marginal value to more extreme cases.
Hence, convexity of the decision problem encourages complementarity; and it cannot discourage complementarity very much because a convex function can only be ``so flat''.
This is pictured in Figure \ref{fig:universal}.
By moving the posteriors $p_{a_1,a_2}$ all very far from the prior, compared to the posteriors $p_{a_1},p_{a_2}$, we get increasing marginal returns.

Pick any convex $G$ on $[0,1]$ and let all probability distributions on $E$ be represented as scalars $q \in [0,1]$.
This includes the prior $p$ and posteriors such as $p_{a_1}$.
It suffices to show that $\E G(p_{a_1a_2}) + G(\bot) \geq \E G(p_{a_1}) + \E G(p_{a_2})$.
Let $p^*$ be the minimizer of $G$ on $[0,1]$.
Then $G(p_{a_1a_2}) \geq G(p_{a_1}) + G'(p_{a_1}) (p_{a_1a_2} - p_{a_1})$.
Now we claim that $G'(p_{a_1})(p_{a_1a_2} - p_{a_1}) \geq G(p_{a_1}) - G(p)$.
This implies $G(p_{a_1a_2}) + G(p) \geq 2G(p_{a_1})$, and after taking the analogous case reversing $a_1$ and $a_2$ and expectations, we get that the desired inequality must hold.
To prove the claim, consider the case $p_{a_1a_2} > p_{a_1}$.
Then $(p_{a_1a_2} - p_{a_1}) \geq r$, and we want to show $G'(p_{a_1}) \geq \frac{G(p_{a_1})-G(p)}{r}$.
This holds by definition of a subgradient as the right side is at most the slope of a line connecting $p$ and $p_{a_1}$.
The analogous argument holds for the case $p_{a_1a_2} < p_{a_1}$.
\end{proof}

\begin{corollary} \label{cor:univ-comps-example}
An example of universal complements are the signals $A_1,A_2$ each independently equal to $1$ with probability $q \in [0.25,0.75]$ and $E = A_1 \oplus A_2$ (the XOR operation).
\end{corollary}
We note that these may be universal complements for larger ranges of $q$ as well, and a very interesting question for future work is to characterize the set of universal complements.
This seems to require more work particularly for $E$ with a larger number of outcomes.
For binary $E$ which are equal to a deterministic function of $A_1,\dots,A_n$, it seems possible to relate complementarity to the \emph{sensitivity} of a function in Boolean analysis.

\subsection{Identifying complements}
Our main result here is to identify the following very broad class of complements.
\begin{proposition} \label{prop:indep-comps-bregman}
Independent signals are strong complements in any decision problem where $G$ has a \emph{jointly convex} Bregman divergence $D_G(p,q)$.
\end{proposition}
\begin{proof}
$D_G(p,q)$ is termed jointly convex if it is a convex function on the domain $\Delta_E \times \Delta_E$ (as opposed to the case where it is convex in each argument separately, for instance).
Assume signals are independent; consider any $A,B$ on the subsets signal lattice (incomparable to each other) and any $A'$ on the continuous lattice with $A' \preceq A$.
By Lemma \ref{lemma:marginal-contribution-bregman}, showing strong substitutes is equivalent to showing that
 \[ \E_{a',b} D_G(p_{a'b},p_{a'}) \leq \E_{a,b} D_G(p_{ab},p_{a}) . \]
Use independence and rewrite to this convenient form:
 \[ \E_{b} ~\E_{a'} D_G(p_{a'b},p_{a'}) \leq \E_{b} \E_{a'} \E_{a|a'} D_G(p_{ab},p_{a}) . \]
Note that if $b$ were not independent of $a,a'$, then the inner expectations would require conditioning the distributions of $a$ and $a'$ on the outcome of $b$.
Now, it suffices to prove the following for all $b,a'$:
 \[ D_G(p_{a'b},p_{a'}) \leq \E_{a|a'} D_G(p_{ab},p_{a}) . \]
Now, since $D_G$ is jointly convex, Jensen's inequality will imply this fact if we can just show that $p_{a'b} = \E_{a|a'} p_{ab}$ and $p_{a'} = \E_{a|a'} p_{a}$.
We prove the first equality; the second is exactly analogous but easier.
\begin{align*}
  p(e \mid a',b) &= \frac{p(e,a',b)}{p(a',b)}  \\
    &= \sum_a \frac{p(e,a,a',b)}{p(a',b)}  \\
    &= \sum_a \frac{p(e|a,a',b) p(a,a',b)}{p(a',b)}  \\
    &= \sum_a p(e|a,b) \frac{p(a,a')p(b)}{p(a')p(b)}  \\
    &= \sum_a p(e|a,b) p(a|a') .
\end{align*}
We used Bayes rule and the law of total probability, then eventually used the fact that $p(e|a,a',b) = p(e|a,b)$ (because $a'$ and $e$ are independent conditioned on $a$) and that $a$ and $a'$ are independent of $b$.
\end{proof}
A corollary is that independent signals are complements for the $\log$ scoring rule and the quadratic scoring rule, as their divergences ($KL$-divergence and $L_2$ distance squared) are jointly convex.

On the other hand, some form of this restriction on the decision problem is needed:
\begin{claim}
If $D_G$ is not convex in its second argument, then independent signals are not necessarily complements.
\end{claim}
\begin{proof}
We consider a binary event $E$ and $G(q)$ where $q \in [0,1]$ is a probability that $E = 1$.
A counterexample is to let $G(q) = 0$ for $q \leq 0.75$ and $G(q) = q-0.75$ for larger $q$.
Consider any decision problem associated with this $G$ (for instance, predicting against a proper scoring rule derived from $G$).
The two signals $A$ and $B$ are independent uniform random bits, and $E$ is equal to the binary OR of the bits.
In this case, one can check that $\V(A) = \V(B) = \frac{1}{8}$, but $\V(\bot) = 0$ and $\V(A \join B) = \frac{3}{16}$.
Thus in particular $\V(A) + \V(B) \geq \V(A \join B) + \V(A \meet B)$, so they are not complements.
Note that $G$ may be modified to be strictly convex while preserving the counterexample.
Also note that the sharp ``kink'' in the graph of $G$ at $q=0.75$ forces $D_G$ to be non-convex in its second argument (one can take the first argument to be the prior on $E$, $q=0.75$, choosing the subgradient $G'_{0.75}(q) = q-0.75$).
\end{proof}

%

\subsection{Identifying substitutes}

It was previously known~\citep{chen2010gaming} for prediction markets under the $\log$ scoring rule that conditionally independent signals induced players to reveal their information as early as possible.
The result in a different guise was shown in an algorithmic context, namely, that for conditionally independent signals, .
Both of these turn out to correspond to the fact that signals are always substitutable in the context of the $\log$ scoring rule:
\begin{theorem} \label{thm:log-CI-subs}
For the $\log$ scoring rule, signals that are conditionally independent upon the event $E$ are strong substitutes.
\end{theorem}
\begin{proof}
Suppose signals $A_1,\dots,A_n$ are conditionally independent on $E$.
Let $A,B$ on the subsets signal lattice and $A'$ on the continuous lattice with $A \meet B \preceq A' \preceq A$.
We are to show that $\V(A' \join B) - \V(A') \geq \V(A \join B) - \V(A)$.
Recall that, using the entropy characterization, $\V(A) = -H(E|A) = H(A) - H(E,A)$ by the properties of the entropy function.
\begin{align*}
  \V(A' \join B) - \V(A')
  &= H(A', B) - H(E, A', B) - H(A') + H(E, A') \\
  &= H(B | E, A') - H(B | A').
\end{align*}
Similarly, $\V(A \join B) - \V(A) = H(B | E,A) - H(B | A)$.
Now because $A' \succeq A \meet B$, by conditional independence, $H(B|E,A') = H(B|E,A)$.
Meanwhile, by conditional independence, $H(B|A') \leq H(B|A)$, which completes the proof.
\end{proof}

However, this fact is apparently quite special to the $\log$ scoring rule.
\begin{proposition} \label{prop:quad-ci-not-subs}
Conditionally independent signals are not necessarily substitutes for the quadratic scoring rule (which has a jointly convex Bregman divergence).
\end{proposition}
\begin{proof}
The quadratic scoring rule has expected score function $G(p) = \|p\|_2^2$, and therefore $S(p,e) = 2p(e) - \|p\|_2^2$.
Let $E$ be binary with $p(E=1) = r > 0.5$ and let each of $A,B$ be i.i.d. conditioned on $E$, each equal to $E$ with probability $s > 0.5$ and equal to the bit-flip of $E$ otherwise.

One can check that for cases where $r$ is large compared to $s$, for instance $r = 0.9$ and $s = 0.8$, we have $\V(A) + \V(B) \leq \V(A \join B) + \V(A \meet B)$ (recall that $\V(A) = \E_a G(p_a)$).
Hence substitutability is strongly violated. For instance, an agent observing $A$ would prefer to report second in a prediction market after an agent observing $B$, even when both agents are constrained to report truthfully.

Intuitively, what happens in this information structure is that neither the realization of $A$ nor of $B$ on its own is enough to change the rather strong prior on $E$.
However, sometimes, observing both $A$ and $B$ (if they are both $0$) can cause a large change in beliefs about $E$, which means that observing both can sometimes be very valuable.
\end{proof}

\subsection{Designing to create substitutability} \label{sec:surrogate}
We now briefly consider the question of designing a decision problem or scoring rule so as to enforce substitutability of information.
In a game-theoretic setting such as prediction markets, one would like to design mechanisms where information is aggregated quickly; as we have seen, this is essentially equivalent to making information more substitutable.

In an algorithmic setting, the decision problem with which one is faced may be difficult to optimize due to non-substitutability of information.
One would like to construct a ``surrogate'' decision problem for which the information at hand is substitutable, then use algorithms for (Adaptive) \sigsel/ to approximately maximize that surrogate, with some guarantee for the original problem.
This is the approach of \emph{submodular surrogates} in the literature~\citep{chen2015submodular}.
We do not directly consider this problem in this paper, but we hope that these techniques may yield insights or tools that are useful in these problems for future work.

\begin{figure}[ht]
\caption{\textbf{Encouraging information to be substitutes.}
  Here, $E$ is binary, and the $x$-axis plots the probability $q$ that $E=1$, with the black curve plotting $G(q)$, the expected score function corresponding to a proper scoring rule.
  Illustrates an informal case with two signals $A$ and $B$, each of which may be ``lo'' or ``hi'', together with some distributions on $E$: the prior, the posterior conditioned on one of the signals being ``lo'' or ``hi'', and the posterior conditioned on both signals being ``lo'' or ``hi''.}
\label{fig:good-bad-G}
\begin{subfigure}{1.0\textwidth}
\centering
\includegraphics[width=0.7\textwidth]{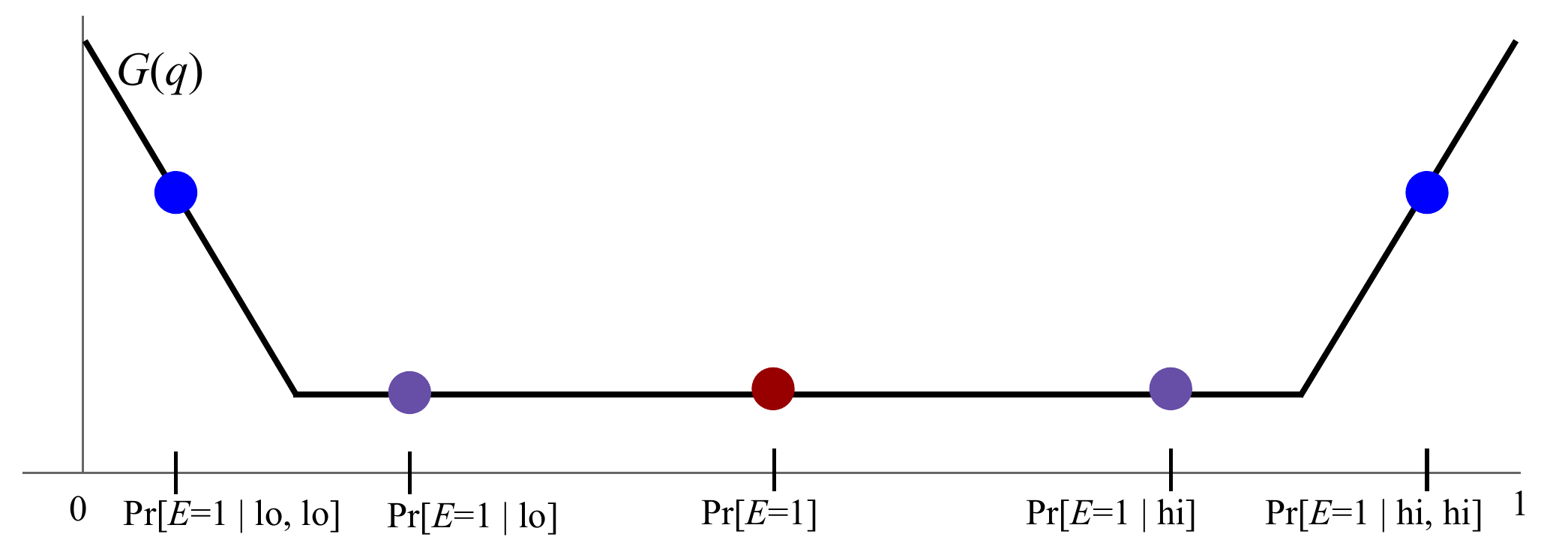}
\caption{\textbf{A ``bad'' choice of $G$.}
  Here, many information structures will be complements because the value of any new information is very small, yet additional information on top of that becomes valuable.
  In particular, for each signal, say $A$, $\E_a G(p_a) = G(p)$ where $p$ is the prior (purple points).
  Yet the expected value of both signals, $\E_{ab} G(p_{ab})$ (blue points), is larger than $G(p)$.}
\end{subfigure}

\vspace{16pt}

\begin{subfigure}{1.0\textwidth}
\centering
\includegraphics[width=0.7\textwidth]{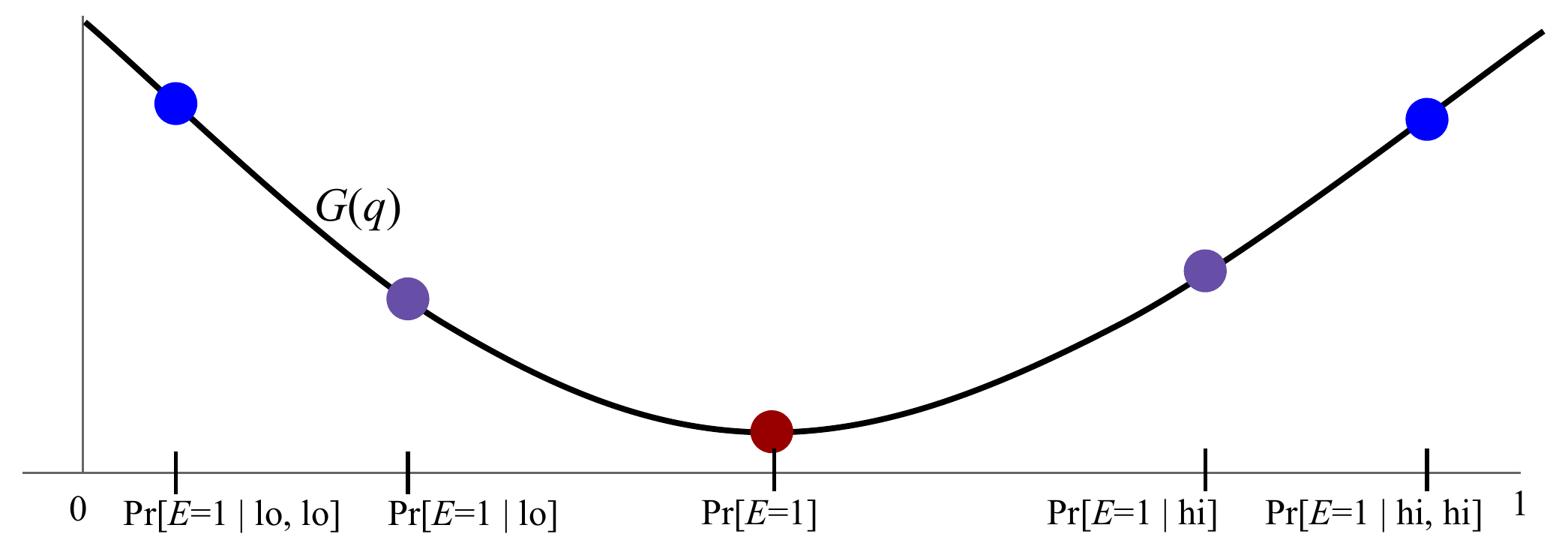}
\caption{\textbf{A ``good'' choice of $G$.}
  Here, many information structures will be substitutes because the value of new information is initially high around the prior.
  One way to see this is that the Bregman divergence becomes low far from the prior; recalling that $\E_{a,b} D_G(p_{ab},p_a)$ is the expected difference in score between predicting $p_{ab}$ versus just predicting $p_a$, this implies that there is a small expected gain from observing $B$ once one knows $A$ (in other words, substitutability).}
\end{subfigure}
\end{figure}

An important conclusion of this paper is that substitutability, or lack thereof, tends to arise from a combination of two factors:
\begin{enumerate}
 \item The \emph{distances between beliefs} due to the information structure.
       If updating on additional information tends to spread beliefs farther apart, then that information tends to be complementary.
       If it tends to make smaller changes in beliefs, it tends to be substitutable.
 \item The \emph{curvature} of the expected utility function $G: \Delta_E \to \R$ associated to the decision problem.
       Highly-curved regions correspond to high marginal value of information to beliefs in those regions.
\end{enumerate}
We illustrate some of this intuition, with an eye toward design of substitute-encouraging $G$, in Figure \ref{fig:good-bad-G}.
By designing a $G$ that is highly curved close to the prior, then gradually less curved farther from the prior, one increases marginal value of information near the prior and decreases it (relatively) further away; this increases substitutability.

\paragraph{A formalization of the problem.}
Given an information structure consisting of a prior $E$ and signals $A_1,\dots,A_n$, (when) can we construct a decision problem $u$ such that the signals are substitutes?
Noting that a trivial decision problem, \emph{e.g.} with one action, technically satisfies substitutes, we will seek decision problems that satisfy a ``nontriviality property''.
\begin{definition} \label{def:somewhat-strict}
Substitutes are \emph{somewhat strict} if the marginal value is always diminishing on the corresponding lattice, and sometimes strictly diminishing.
Analogously, complements are \emph{somewhat strict} if marginal value is always increasing and sometimes strictly increasing.
\end{definition}
Unfortunately, our results on universal S\&C in Section \ref{sec:universal} essentially imply that this is not always achievable.
\begin{proposition} \label{prop:always-design}
For some information structures, it not possible to design a decision problem giving rise to weak, somewhat strict substitutes; the same holds for complements.
However, if the information structure does not contain a mixture of trivial substitutes, then it is possible to design a decision problem under which signals are weak, somewhat strict complements.
\end{proposition}
\begin{proof}
We showed in Corollary \ref{cor:univ-comps-example} that nontrivial universal complements exist; for these structures, for every decision problem, they are weak complements, which implies that they cannot be somewhat strict substitutes.

Meanwhile, in Proposition \ref{prop:subs-mix-trivial}, we showed that unless the information structure contained a mixture with trivial substitutes, one can construct a decision problem satisfying weakly and sometimes strictly increasing marginal value.
\end{proof}
A crucial direction for future work is to better characterize for what information structures it is always possible to design for substitutability.

%

.

\clearpage

\section{Discussion, Conclusion, and Future Work} \label{sec:conclusion}
\subsection{Contributions and discussion}
This paper makes two types of contributions.
One type is to propose definitions of informational substitutes and complements and provide evidence that these definitions are natural, useful, and tractable.
The other type is to prove concrete results for game-theoretic and algorithmic problems.
These two contributions are interrelated: The concrete results rely on the definitions and characterizations of S\&C, while evidence in favor of the definitions comes from their utility in proving the concrete results.

We begin with a summary of the contributions in terms of concrete results, then summarize the evidence in favor of our proposed definitions of informational S\&C.
We then discuss future work in a variety of directions.

\paragraph{Contributions independent of S\&C}
The first application in this paper was to identify essentially necessary and sufficient conditions for the two types of prediction market equilibria that have been most studied: The ``good'' case where all traders rush to reveal information as soon as possible, and the ``bad'' case where all traders delay all revelation as long as possible.
This result broadly generalizes previously-known special cases~\cite{chen2010gaming,gao2013jointly}.
These conditions corresponded respectively to our definitions of informational substitutes and complements.
Our other results regarding these definitions have implications and applications for the market setting.
We gave some tools and approaches for identifying and designing substitutability, which (we showed) corresponds to the main goal of a market designer: encouraging immediate information revelation.

We also gave some additional game-theoretic applications in other settings involving strategic information revelation and aggregation.
We hope that future work can explore more connections.

The second main application was to the algorithmic problem of information acquisition, which we formalized as the \sigsel/ problem.
This is a very natural and general problem, capturing any sort of information acquisition scenario.
We showed that substitutes imply efficient approximation algorithms via submodular maximization, offering a unifying lens or perspective on a variety of literature utilizing this approach.
We also saw matching hardness in general when substitutability is not present.
These results give hope that substitutable structure may be leveraged in future algorithmic investigations of related problems.

\paragraph{Evidence for informational S\&C}
Informational substitutes require much more background and work to define than their counterparts for goods, as discussed in Section \ref{sec:develop}.
This leads to many ways in which the definition can ``get it wrong''.
However, there are several reasons to believe that the definitions proposed in this paper are substantially on the right track.
The three goals we address are showing that they are \emph{natural}, \emph{useful}, and \emph{tractable}.
The evidence for these is summarized below; however, there is room for further investigation and discussion on the subtleties of the proposed definitions.
As additional future applications are investigated, it may be found that tweaks in the definition improve it along these axes.

\paragraph{Natural.}
The following evidence suggests that the definitions of informational S\&C presented in this paper are natural.
\begin{itemize}
 \item They are defined in terms of sub- and super-modular lattice functions, sharing this property with widely accepted definitions of S\&C for items.
       The weak, moderate, and strong versions of the definitions correspond to lattices have natural interpretations both game-theoretically -- strategies that are binary, deterministic, or randomized -- and algorithmically, optimizing over subsets, deterministic ``poolings'', or randomized ``garblings''.
 \item They can be alternatively characterized information-theoretically, by diminishing (increasing) amount of information revealed by a signal.
 \item They can be alternatively characterized geometrically, by the distance a belief is moved by a posterior update on one signal given the other.
 \item S\&C respectively characterize ``all-rush'' and ``all-delay'' equilibria in prediction markets.
       Hence, our definitions (or very close variants) seem unavoidable when studying equilibria of these basic models of strategic play.
 \item The algorithmic complexity mirrors the now-familiar story in the case of items: approximately optimizing over substitutes can be done efficiently (at least in the case of weak substitutes), while for the general case or complements the problem is difficult.
\end{itemize}

\paragraph{Useful.}
The usefulness of the definitions, so far, is reflected by the results summarized above.
In particular, they allowed resolving an open problem on strategic aggregation in prediction markets that was previously solved only for particular cases.

\paragraph{Tractable.}
Our definition of informational S\&C refers to a very general setting of decision problems and information structures.
It is not initially clear that a definition that broadly captures S\&C can also be amenable to analysis or intuition.
This particularly presents a challenge for our definitions because they depend on both the decision problem and the information structure, each of which can be a very complicated object.

We presented a convexity-based approach, namely, studying the convex \emph{expected-utility} function $G$, and showed that it both captures geometric intuitions for understanding S\&C and gives useful analytic tools.
One of these is prediction and the theory of proper scoring rules, which gives a way to construct a decision problem from any such convex $G$.
We used these tools to give some example broad classes of informational S\&C as well as positive and negative results on designing decision problems to encourage substitutability or complementarity.
We also gave intuitive definitions of informational S\&C from the perspectives of information theory (generalized entropies) and geometry (divergences).
These found some formal applications already in this paper, \emph{e.g.} a convexity condition on the divergence definition of substitutes implies that all independent signals are substitutes.

\subsection{Future work: game theory}
An immediate direction is to extend our results to broader models of financial markets.
We believe that analogous results are likely to hold.
It would also be ideal, albeit esoteric, to understand the distinction between when Bayes-Nash equilibria and perfect Bayesian equilibria exist or are essentially unique.
In a Bayes-Nash equilibrium that is not perfect Bayesian, some trader essentially makes a ``threat'' that is ``not credible''.
Can such scenarios exist when signals are substitutes or complements?

Many broader questions about S\&C have a direct implication on markets via our results, and we outline some of these directions below.

For more general Bayesian games, the implications of our results are not yet clear.
With multiple players, for instance, it is no longer true that more information always helps.
It may be most tractable to look at special cases such as auctions or signaling games, although those seem to have a significant component of \emph{strategic} substitutes as well.
One natural direction is common-value auctions, where in one case~\citep{milgrom1982value} informational substitutes have been explicitly used.

\subsection{Future work: algorithms}
There are many potential algorithmic questions raised by these definitions.
One is to consider existing problems, such as algorithmic Bayesian persuasion~\citep{dughmi2015algorithmic} or signalling in auctions, through the lens of substitutability.
However, such problems are already significantly more complex because of the interaction between the players, which is largely absent or abstracted out so far in this work.
We hope that further investigation will uncover the connections to such settings.
Meanwhile, one might ask whether, for instance, variants of \sigsel/ on the discrete or continuous lattice are well-motivated in any settings and, if so, whether they are tractable.

More specifically and technically, one might ask whether the oracle model presented in this paper the best way to represent a decision problem, or if there is some other natural alternative beyond those we presented.
Perhaps it is possible to obtain positive results with weaker or substantially different assumptions on the oracles.

Another group of algorithmic questions relates to better understanding the definitions themselves.
One concrete question is the following: Given a decision problem, perhaps in the oracle model (or some other well-justified model), determine whether signals are substitutes, complements, or neither.
Or more strongly, given a decision problem and two signals $A,B$, find
 \[ \arg\min_{Q \preceq A} \V(Q \join B) - \V(Q) . \]
An algorithm for this problem and multi-signal generalizations would help identify substitutes or complements, and would more generally solve \emph{e.g.} the problem of how to trade in general prediction markets where signals are neither substitutes nor complements, but something in between.
Many more possible algorithmic questions are likely to arise as investigation continues.

\subsection{Future work: structure of S\&C}
A straightforward direction for future work is to identify further classes of substitutes and complements, including natural classes of decision problems and information structures that together produce S\&C.
This question also bears directly on the design of prediction markets.

The second straightforward direction is how to ``design for substitutability''.
We believe that our results give a good start in this direction, but many questions remain.
For instance, can we characterize all universal substitutes and complements? (This characterizes cases where we cannot design for substitutability or complementarity.)
And, in cases where we can make signals strict substitutes, can we design some natural and computationally efficient method for doing so?
We hope that the geometric intuitions in this paper provide a starting point for addressing such questions.

Finally, one could be interested in exploring variants of the definitions taking the same approach but with some piece of the puzzle substantially altered.
Concretely, one could imagine a very risk-averse agent with an adversarial view of nature and some more combinatorial representation of information.
Can one formulate analogous definitions in cases like these?

\clearpage

\bibliographystyle{plainnat}
\bibliography{citations}

\clearpage

\appendix
\section{Other Related Work} \label{sec:related}
We first survey related work on substitutes and complements for information (and in general).
Then, we discuss work relating to information aggregation in markets and other game-theoretic settings.
Finally, we describe algorithmic work on information acquisition (particularly in ``submodular'' settings).

\subsubsection{Substitutes and complements}
The notion that pieces of information may exhibit substitutable or complementary features is certainly not a new intuition; but up until this work, it seems to have remained mainly an intuition.
There seem to be few attempts at formalizing a general definition or even special cases.
The only work we know of in this direction is \citet{borgers2013signals}, which inspires our approach but also has significant drawbacks and limitations.
We extensively discuss \citet{borgers2013signals} in in Section \ref{sec:develop}, where we contrast it with our definitions.

Two works in (algorithmic) game theory touching on informational S\&C are \citet{jain2009designing} and \citet{milgrom1982value}; however, these do not propose general definitions.
They are discussed below.

Somewhat related is the game-theoretic notion of \emph{strategic substitutes and complements}~\citep{bulow1985multimarket}.
Roughly, these concepts refer to cases where a change in action by one player in a game results in a response from another that is similar to the first player's (in the case of complements) or offsetting (in the case of substitutes).
This notion seems relatively unrelated to our definitions of informational S\&C.
For one, our definitions focus on the case of a single decisionmaker or single optimization problem.
Also, strategic S\&C can be defined in complete-information games, where there are no signals or information of any kind.
However, perhaps future work can discover classes of games in which the notions are more closely related.

\paragraph{Valuations for items.}
In contrast to the lack of literature on informational S\&C, in the case of valuation functions for \emph{items}, substitutability and complementarity have been put on firm formal foundations, with strong connections between substitutes and existence of equilibria in markets for goods or matching markets~\citep{kelso1982job,roth1984stability,gul1999walrasian,hatfield2005matching,ostrovsky2008stability}.

In computer science, the literature on optimization has produced strong positive results leveraging substitutable structure.
The main example of this is submodularity, which has been connected to computational tractability throughout theoretical computer science and machine learning~\citep{krause2012submodular}.
Submodularity, in addition to having nice algorithmic properties, is also recognized as a natural model of substitutes in (algorithmic) game theory~\citep{lehmann2001combinatorial}.

This paper draws parallels to such research because our definition of informational substitutes turns out to correspond formally to submodular valuation functions.
We show market equilibrium results of a similar flavor, but for ``information markets''; and we also show algorithmic results of a similar flavor for the problem we call \sigsel/, which is the problem of selecting an optimal set of signals subject to constraints.

However, we emphasize that informational S\&C pose challenges that do not arise in the item setting:
\begin{itemize}
 \item Items are modeled as having an \emph{a priori} innate value.
       Information is not; its value must arise from context.
 \item Items are modeled as being \emph{atomic} or indivisible, with no inner structure.
       In contrast, information, modeled as \emph{e.g.} random variables, is defined by inner structure: the probability distributions from which it is drawn.
 \item The relationship between items is completely determined by the valuation function in the context of interest.
       Concretely, when modeling a set function $f:2^{\{1,\dots,n\}} \to \R$, it is usually not the case in the model that items $3$, $7$, and $k$ have a special relationship that has an impact on allowable forms of $f$.
       In contrast, in a value-of-information setting, the value of observing a triple of signals cannot be completely arbitrary; it must depend somehow on correlations between these signals.
\end{itemize}
We give an in-depth description of how our definitions overcome these challenges in Section \ref{sec:develop}.

\subsubsection{Information in markets}
The ``efficient markets hypothesis'' (EMH) refers to a large set of informal conjectures about how quickly information is revealed and incorporated into the prices of financial instruments in markets.
For our purposes, a ``financial instrument'' may be formalized as a \emph{security} with an uncertain value, which will be revealed after the close of the market; each share purchased of that security may then be redeemed for a payout equal to this revealed true value of the security.
Concretely, one can picture a binary security that has value $1$ if a certain event occurs and $0$ otherwise.

\citet{fama1970efficient} discussed formalizations of the EMH with varying levels of strength.
\citet{kyle1985continuous} defined a formal model of financial traders in markets, involving both informed traders and ``noise'' traders who are uninformed and essentially trade randomly; this is the current most common model of such trading in the economics literature.
However, formal progress on this question was very slow until \citet{ostrovsky2012information} showed that, in equilibrium of this model, information is always \emph{eventually} aggregated under a certain condition on securities.
Formally, this means that the price of a security converges to its \emph{ex post} expected value conditioned on the information held by all traders.
\citet{ostrovsky2012information} considered both finite-round and infinite-round markets (with and without discounting), and considered \emph{prediction markets} (described below) as well as Kyle's model.
One subtlety that may be worth pointing out is that, in an infinite-round market without discounting, it is not known that a Bayes-Nash equilibrium always exists, while this is known for the other cases.
\citet{ostrovsky2012information} showed in all cases that aggregation occurs in any equilibrium, under his ``separability'' condition on securities.
The condition essentially ensures that, if information is not yet aggregated, then some participant has information that can be used to make a ``useful'' trade, \emph{i.e.} one that makes money and (therefore) intuitively makes ``progress'' toward aggregation.

This showed that information is eventually aggregated; but \emph{how} is it aggregated?
It would be ideal if traders rush to reveal their information, but very bad if they ``delay'' as long as possible.

Such questions are difficult to address in the economic model of financial markets of \citet{kyle1985continuous,ostrovsky2012information}.
Research on the dynamics of strategic trading has made some progress in the model of prediction markets.
These are simplified financial markets in which there are no uninformed ``noise'' traders and in which participants generally interact one-at-a-time with a centralized market maker, who sets prices via a transparent mechanism and subsidizes the market.
Specifically, research on strategic play focuses, as does this paper, on the \emph{market scoring rule} design of prediction markets~\citep{hanson2003combinatorial}, which are based on proper scoring rules (see \emph{e.g.}~\citet{savage1971elicitation,gneiting2007strictly}).
However, we note that market scoring rule markets are equivalent, in a strategic sense, to more traditional-looking ``cost function'' based prediction markets~\cite{abernethy2013efficient}.

The following was known about equilibrium in such markets prior to this work, in addition to \citet{ostrovsky2012information}.
\citet{chen2007bluffing} studied the $\log$ scoring rule and a particular type of information structure among the traders, namely, that each trader's ``signal'' (information) was distributed independently of all others' conditional on the true value of the security.
(For a simple example of such a structure, suppose that the true value of the security is distributed randomly in some way; then each trader observes this true value plus independent noise.)
In this conditionally-independent case, \citet{chen2007bluffing} showed that the ``ideal'' outcome does indeed occur in equilibrium: Traders all rush to reveal their information as early as possible.

Subsequently, \citet{dimitrov2008non} also studied the $\log$ scoring rule but considered other signal structures, particularly independent signals (that is, unconditionally independent). 
For a simple example, suppose each trader observes an i.i.d. random variable, and the true value of the security equals the sum of all the traders' observations.
\citet{dimitrov2008non} showed that, in this case, the ``ideal'' outcome does not occur.
However, when assuming discounting and an infinite number of trading rounds, \citet{dimitrov2008non} showed that information is ``eventually'' aggregated.
This result was generalized to any scoring rule (not just $\log$) by \citet{ostrovsky2012information}.
\citet{chen2007bluffing} and \citet{dimitrov2008non} were combined and extended in \citet{chen2010gaming}.

Then, \citet{gao2013jointly} revisited the $\log$ scoring rule in finite-round markets and considered the information structure where all traders signals are unconditionally independent.
In this case, \citet{gao2013jointly} showed that the ``worst possible'' outcome occurs in equilibrium: Traders all delay as long as possible before making any trades based on their information.
This casts doubt on the efficient markets hypothesis and suggests, taken in tandem with \citet{chen2007bluffing}, that structure of information is crucial to determining strategic behavior.

\subsubsection{Other game-theoretic settings.}
In almost any Bayesian extensive-form game, the question of information revelation is relevant.
While we hope that future work may expand the set of topics to which informational substitutes and complements may apply, in this section, we will focus on the most closely related works, those that directly touch on S\&C, or those for which we show results in Section \ref{sec:game-other}.

The model of prediction markets is in some sense the simplest model of strategic information revelation in dynamic settings.
Thus, it is natural that settings such as crowdsourcing contests are closely related.
One such model of crowdsourcing and machine-learning contests appears in \citet{abernethy2011collaborative,waggoner2015market}.
In that framework, participants iteratively provide data sets or propose updates to a central machine-learning hypothesis, being rewarded for the improvement they make to performance on a test set of data.
A prediction market can be seen as a special case.
Those works did not address strategic equilibria of the mechanisms they proposed.

Another related setting is the model of question-and-answer forums in \citet{jain2009designing}.
That paper introduced a model where an asker has some value function for ``pieces of information'', which are not modeled directly.
Participants can strategically choose when to reveal information.
\citet{jain2009designing} identified ``substitutes'' and ``complements'' cases in which participants rush (respectively, delay) to provide answers.
However, \citet{jain2009designing} did not provide any endogenous model for asker utility or information; the information was modeled almost as discrete items without structure.
Hence, it was not clear from that work under what circumstances (if any) pieces of information would satisfy their substitutes and complements conditions.

In Section \ref{sec:game-other}, we describe implications of our work for results in the above two settings.

More broadly, there are large literatures dealing with signalling in games~\citep{spence1973job}, or the more recent Bayesian persuasion literature~\citep{kamenica2011bayesian,gentzkow2011competition,dughmi2015algorithmic}.
While the models and questions in this area are related, to our knowledge there is no immediate connection.
It may be that future work uncovers connections of informational S\&C to this field, but the literature on persuasion and signalling in games does not seem to have developed notions of informational substitutes nor tools for addressing the applications considered in this paper.

Another significant line of work has considered signalling in auctions, \emph{e.g.} \citet{milgrom1981rational,milgrom1982value}.
The only paper in this literature that we know to explicitly formalize a notion of substitutable information is \citet{milgrom1982value}, which considers a common-value auction with two bidders, one informed and one uninformed, with two signals, each a real-valued random variable with some positive affiliation with the item's value.
The authors define a notion of substitutes specific to their context and show that it implies intuitive properties for this asymmetric-information auction setting.
In future work, it would be interesting to see if there is a formal connection of our definitions to their setting.

\subsubsection{Algorithms for information acquisition.}
The value of information to a decision problem was formally introduced by \citet{howard1966information} and is also closely related to the classic problem in statistics of Bayesian experimental design~\citep{lindley1956measure}.
Given this perspective, it is natural to consider the problem of acquiring information under constraints.
This problem has historically been investigated from many different angles, \emph{e.g.}~\citep{mookerjee1997sequential}.
It is known to be very computationally difficult in some general settings~\citep{krause2009optimal}.

A successful recent trend in this area is to leverage submodular structure to apply efficient approximation algorithms.
For instance, an approximation ratio of $(1-1/e)$ is obtained by the greedy algorithm for maximizing a monotone submodular function subject to a cardinality constraint.
This algorithm or related submodular maximization algorithms were utilized by \citet{krause2005optimal,guestrin2005near}, and a variety of literature since; see \citet{krause2012submodular} for a survey.
In cases where the information is acquired not in a batch but adaptively over time, based on the information observed so far, the problem (and/or solution) is known as \emph{adaptive submodularity}~\citep{asadpour2008stochastic,golovin2011adaptive}.

\section{Defining S\&C: intuition, challenges, and historical context} \label{sec:develop}
In this section, we describe the intuition, justification, and historical context behind our proposed definitions of informational substitutes and complements (S\&C).
The formal definitions are presented more concisely starting in Section \ref{sec:defs}, to which the reader may skip if uninterested in background.
We focus on the substitutes case; the complements case, where not mentioned, is analogous.

The only prior attempt at definitions of S\&C of which we are aware is \citet{borgers2013signals}, which will be introduced shortly.
The setting and approach in that paper are very appealing, and they inspire our approach in this paper.
However, there are also key drawbacks that motivate us to diverge from their approach in several ways.
We will describe in context the drawbacks and our approach to overcoming them.

\subsubsection{Defining the value of information}
Defining informational S\&C turns out to be significantly more work than defining substitutes and complements for valuation functions over items, starting from the very beginning: How does \emph{value} arise in the first place?
It is generally accepted to model a valuation function over a set $U$ of goods as some $f: 2^U \to \mathbb{R}$, without justifying how $f(S)$ arises (perhaps the items are yummy, shiny, or have other desirable qualities).
However, outside of curiosity, it seems that information's innate value is more questionable; and furthermore, should depend on its probabilistic structure.
For instance, two signals may be independent or may be highly correlated.
How does this relate to their value?

A solution arises from the observation that information's value is often determined by the \emph{use} to which the information may be put.
As in \cite{howard1966information,borgers2013signals}, for us the value of information arises from its utility in the context of a decision problem.
We consider a Bayesian model of information in which there is a prior distribution on the event $E$ of interest and on the possible pieces of information, called \emph{signals}. 
In a decision problem, the agent observes some signals and then makes a decision $d \in \D$, after which the outcome $E=e$ is revealed and the agent's utility is $u(d,e)$.
Thus (following \cite{borgers2013signals}), we define a \emph{value function} $\V^{u,P}$ on signals, for a given decision problem $u$ and prior $P$, by the expected utility of the agent given that she first observes $A$, then takes the optimal action based on that information. Formally,
 \[ \V^{u,P}(A) = \E_{a\sim A} \left[ \max_{d\in\D} \E_{e\sim E} \left[ u(d,e) \mid A=a\right] \right] . \]
Thus, one can compare, for instance, the value $\V^{u,P}(A_1)$ for observing signal $A_1$ alone versus the value $\V^{u,P}(A_2)$ for observing signal $A_2$ alone.
Note that the \emph{marginal value} for observing signal $A_2$, given that the agent will already observe signal $A_1$, is $\V^{u,P}(A_2 \join A_1) - \V^{u,P}(A_1)$, where the notation $A_2 \join A_1$ means to observe both signals (this notation will be explained shortly).

\subsubsection{The approach of B\"{o}rgers et al.}
\citet{borgers2013signals} proposes the following definition: Given an event $E$, two signals $A_1$ and $A_2$ are substitutes if for \emph{every decision problem} (and associated value function $\V^{u,P}$),
 \[ \V^{u,P}(A_1) + \V^{u,P}(A_2) \geq \V^{u,P}(A_1 \join A_2) + \V^{u,P}(\bot) , \]
where $A_1 \join A_2$ denotes observing both signals, while $\bot$ denotes not observing any signal and deciding based on the prior.

This definition has two properties that might seem attractive, but turn out to be fatal in many cases of interest: (a) it does not depend on the particular decision problem, but only on how $A_1$, $A_2$, and $E$ are correlated; (b) it depends only on the values $\V^{u,P}(A_1)$, $\V^{u,P}(A_2)$, of both, and of neither, and does not depend on any internal structure of $A_1$ and $A_2$.
We explain why these properties are problematic and how our definition will differ.

\paragraph{a. Lack of dependence on the decision problem.}
The problem here is that in a majority of cases, two signals can turn out to be either substitutes or complements depending on the decision problem at hand.
For example, whether two weather observations are substitutes or complements depends on what decision is being made.
Temperature and dew point might be considered complements when deciding whether to bring one's umbrella.\footnote{Accepting the proposition that knowledge of both gives a much better prediction of rain than knowledge of either alone.}
But when deciding, for instance, how warmly to dress, these two measurements might be considered substitutes since, given one of them, the other gives relatively less information about the comfort level of warm or cool clothing.
For another example: To a trader deciding whether to invest in a technology index fund (that is, a stock whose value follows that of the general tech sector), the share prices of two given tech companies may be substitutable information, since each gives some indication of the current value of tech stocks.
But to a trader deciding which of these two specific companies to invest in, these prices may be complements, since the decision can be made much more effectively with both pieces of information than with either alone.

The definition of \cite{borgers2013signals} cannot capture such scenarios because it requires two signals to be substitutes for \emph{every} possible decision problem.
Our solution is to define S\&C relative to both the particular information structure and the particular decision problem.

\paragraph{b. Lack of dependence on the internal structure of the signals.}
The other concern with the definition of \citet{borgers2013signals} is that it only depends on ``extreme'' values: $\V^{u,P}(A_1), \V^{u,P}(A_2), \V^{u,P}(A_1 \join A_2),$ and $\V^{u,P}(\bot)$. Hence, it ignores the internal structure of $A_1$ and $A_2$, which can lead to incongruous predictions.
For example, suppose that $B_1$ and $B_2$ are substitutes while $C_1$ and $C_2$ are complements.
Now consider the signals $A_1 = (B_1,C_1)$ and $A_2 = (B_2,C_2)$.
For some decision problems, it may be that the $B$ signals are slightly more important and so $A_1$ and $A_2$ seem to be substitutes.
For other decision problems, it may be that the $C$ signals are slightly more important and $A_1$ and $A_2$ seem to be complements.
This is formalized in Example \ref{ex:weak-but-not-moderate}.

To see why this could be problematic for a predictive or useful theory, suppose that an agent will be able to observe $A_1$, and a seller wishes to sell to that agent the opportunity to observe $A_2$ as well.
As just argued, one might have defined $A$ and $B$ to be ``substitutes'' or ``complements'' depending on very small fluctuations in the decision problem.
But the seller, by ``hiding'' or ``forgetting'' either the $B_2$ or the $C_2$ component of his signal, can force the signals to become either substitutes or complements as she desires.
A definition that does not account for internal structure may get such examples completely wrong, \emph{e.g.} classifying the signals as substitutes when the seller can make them behave as complements.

We will introduce definitions that account for the internal structure of signals.

\subsubsection{Our approach: dependence on context and internal structure}
\paragraph{Context.}
As mentioned above, we will allow the definitions of S\&C to depend on the particular decision or value function $\V^{u,P}$.
That is, while \citet{borgers2013signals} defined a particular information structure $P$ to be substitutes on pairs of signals if $\V^{u,P}$ satisfied a condition for all $u$, we will define a pair $(u,P)$ to be substitutes if $\V^{u,P}$ satisfies a similar condition.
This will turn out to be crucial in all of our applications.

The potential drawback is that it might be difficult to say anything \emph{general} about when signals are substitutes or complements; it might seem that one must take things completely on a case-by-case basis.

We make two counterpoints.
First, a universal approach may be the wrong goal or ``too much to hope for''.
For instance, in the case of items, there is no such optimistic analogue; one does not consider items that are always substitutes for every valuation function.
Despite this, there are many successful theories leveraging substitutable goods.
These approaches start by assuming a context (\emph{e.g.} valuation functions) for which the goods are substitutes; similarly, we can consider a set of signals and only those decision problems for which they are substitutes.

Second, we later give some evidence that we need not take things completely case-by-case.
We seek classes of signals that can be considered substitutes or complements in a broad class of decision problems.
For example, we show that if signals are independent, then they are complements for \emph{any} decision problem satisfying a smoothness condition.
Our work also gives intuition for which kinds of signals are more likely to be substitutable or are substitutable in more contexts.
And indeed, one of the exciting questions raised by our work is how the context of a decision problem and internal structure combine to produce substitutable or complementary features.

\paragraph{Probabilistic structure.}
We will allow definitions of S\&C to depend on the internal structure of signals.
But how?
Two signals may be related in complex ways by correlations with each other and with the event $E$ of interest.
Therefore, a more natural analogy than substitutability of two items may be substitutability of two \emph{subsets} of items.
Consider \citet{lehmann2001combinatorial}, which studied valuation functions over sets of items.
There, the authors identified a natural ``no complementarities'' condition where two sets of items, $A_1$ and $A_2$, could only be substitutes if all ``pieces'' of those subsets were substitutes: no subset of set $A_1$ could be complementary to the set $A_2$, and vice versa.
This turned out to be exactly a requirement that the valuation function be \emph{submodular}: that it exhibit diminishing marginal value.

We would also like to capture diminishing marginal value.
The challenge that arises is, what is a marginal ``unit'' of information?
The answer actually may vary by application.
\begin{enumerate}
 \item In some applications, a ``marginal unit'' may be an entire signal: Given the current subset of $\{A_1,\dots,A_n\}$, one can either add another $A_i$, or not.
       This would be appropriate for cases where our above arguments about internal structure may not apply.
       For example, perhaps the seller in an auction does not have the ability, for whatever reason, to process her signals in any way; she can just choose between allowing each of them to be either broadcast or kept private.
       In this paper, we utilize this notion, which will correspond to ``weak substitutes'', in the context of discrete optimization problems where an algorithm must choose between acquiring different signals.
       In many contexts, it is impossible to acquire partial signals, so this is the natural marginal unit.
       While they may be useful, they also are subject to the criticisms given above; in many contexts that allow ``pieces'' of signals, they may not behave as substitutes or may even behave as complements.
 \item Sometimes, a ``marginal unit'' may be some partial information about a signal, in the form of a ``fact'' about its realization.
       For instance, imagine a commuter learns something about the barometer reading but not the exact reading; \emph{e.g.}, whether it is above or below $30$, or the measurement rounded to the nearest integer.
       This application arises when considering pure strategies in a game, or deterministic ``post-processings'' of a signal in algorithmic contexts.
       The effect of such processing is always to ``coarsen'' a signal by pooling multiple outcomes together under one announcement.
       In the barometer example, all realizations of the signal below $30$ map to the same result, and all realizations above map to the other result; similarly when rounding to the nearest integer.
       If a set of signals exhibits diminishing marginal value with respect to this notion, we will term them ``moderate substitutes''.
       We actually do not provide an application for moderate substitutes in this paper, but expect them to be useful in contexts such as those just described.
 \item Finally, a ``marginal unit'' may be partial information in the form of a noisy ``signal about the signal''.
       For instance, the commuter may learn the barometer reading plus Gaussian noise.
       To see that this notion may be much more fine-grained than the previous one, imagine starting from the binary barometer example, where the commuter learns whether or not the barometer is below $30$; and now imagine that, with some probability $p$, this observation is ``flipped'' from the correct one.
       When $p=0$, the commuter can be certain that she learns correctly which outcome is the case (above or below $30$).
       But as $p \to \frac{1}{2}$, the commuter learns less from the signal.
       If a set of signals exhibits diminishing marginal value even with respect to such partial information -- for instance, diminishing marginal value as $p$ decreases from $\frac{1}{2}$ to $0$ -- then we term them ``strong substitutes''.
       In applications where, for instance, the barometer observation is controlled by a strategic agent whose strategy consists of a ``garbling'' of that observation, this will be a useful notion of marginal information.
\end{enumerate}
We formalize these marginal units of information using \emph{lattices} of signals: sets of signals with a partial order $\preceq$ corresponding to ``informativeness'' and satisfying some natural conditions.
While our proposed uses for them here are quite new, the lattices we use, or closely related concepts, are relatively classical.
For weak substitutes, we consider the lattice of subsets of signals, with partial order given by set containment; this corresponds directly to subsets of goods and the notion of substitutes is essentially the same.

For moderate substitutes, we utilize a variant of Aumann's classic model of information in Bayesian games~\citep{aumann1976agreeing}, in which signals correspond to partitions on a ground state of the world.
To our knowledge, although it is known that Aumann's signals form a lattice (because the space of partitions do), they have not been used to formalize marginal units of information.
One difference in the variant we propose is that the ground states only determine the signals, not the event $E$ or any other pieces of information; this makes our model much more useful for formalizing marginal pieces of information because the ground states only contain information about the signals.

For strong substitutes, we extend Aumann's model to capture randomized ``garblings'' of signals.
Although this is not the model normally used in that context, the idea and intuition is extremely similar to Blackwell's criterion or partial ordering on signals~\citep{blackwell1953equivalent}.
One major difference is that in our model, there is a particular event $E$ of interest and signals are ordered according to informativeness \emph{about that event}, rather than pure informativeness.
Also, the use of Aumann's partition model allows our signals to form a lattice.

\subsubsection{Capturing ``diminishing'' marginal value}
Luckily, once we have placed a lattice structure on signals, we can apply a now-classic criterion for diminishing marginal value: submodularity.
Often, submodularity is a condition for functions on subsets, \emph{e.g.} $f: 2^{\{1,\dots,n\}} \to \R$, which is submodular if an element $i$'s ``marginal contribution'' to $S$, $f(S \cup \{i\}) - f(S)$, is decreasing as elements are added to $S$.
This is a widely-used model for substitutability of discrete, indivisible items~\citep{lehmann2001combinatorial}.
The same goes for supermodularity, increasing marginal value, and complementary items.

The final piece of our puzzle will be to utilize submodularity and supermodularity, and extensions, to capture diminishing marginal value and increasing marginal value respectively.
This is formalized in Section \ref{sec:defs}.

\end{document}